\documentclass[11pt, a4paper]{article}

\usepackage{graphicx}
\usepackage[ruled, vlined, longend, linesnumbered]{algorithm2e}
\DontPrintSemicolon
\SetKwInOut{Input}{Input}\SetKwInOut{Output}{Output}
\newlength{\commentWidth}
\usepackage{tikz}
\usetikzlibrary{arrows,automata, shapes.geometric, positioning}
\usepackage{amsmath, amsthm}
\usepackage{amsfonts}
\usepackage{xspace}
\usepackage{environ}
\usepackage{stmaryrd}
\usepackage{wasysym}
\usepackage{comment}
\usepackage{dsfont}
\usepackage{microtype}

\DeclareFontShape{OT1}{cmr}{bx}{sc}{<-> cmbcsc10}{}
\usepackage[margin=2.5cm]{geometry}
\usepackage{enumerate}
\usepackage{url}
\usepackage{hyperref}
\usepackage{thmtools}
\usepackage{thm-restate}
\newcommand{\ones}{\mathds{1}}

\NewEnviron{cproblem}[1]{%
\begin{center}\fbox{\parbox{5.7in}{%
    {\centering\scshape #1\par}%
    \parskip=1ex
    \everypar{\hangindent=1em}%
    \BODY
}}\end{center}}
\NewEnviron{cproblemc}[1]{%
\begin{center}\fbox{\parbox{2.6in}{%
    {\centering\scshape #1\par}%
    \parskip=1ex
    \everypar{\hangindent=1em}%
    \BODY
}}\end{center}}
\NewEnviron{cproblemch1}[1]{%
\begin{center}\fbox{\parbox{2.4in}{%
    {\centering\scshape #1\par}%
    \parskip=1ex
    \everypar{\hangindent=1em}%
    \BODY
}}\end{center}}

\DeclareMathOperator{\dkshop}{DKSH}
\DeclareMathOperator{\mcdshop}{MCD}
\newcommand{\dksh}[1]{\dkshop^*_{#1}}
\newcommand{\mcdsh}[1]{\mcdshop^*_{#1}}
\newcommand{\MCDSH}[1]{\textsc{Multicld-Edge Densest-Sub-\(#1\)-hypergraph}\xspace }
\newcommand{\DKSH}[1]{\textsc{Densest-\(k\)-Sub-\(#1\)-hypergraph}\xspace }

\DeclareMathOperator{\lrds}{LRDS^*}
\DeclareMathOperator{\dks}{DKS^*}
\DeclareMathOperator{\E}{E}
\DeclareMathOperator{\NoSM}{NoSM_{\MU, \NU}}
\DeclareMathOperator{\NoSMs}{NoSM}
\newcommand{\LRDS}{\textsc{LR-Densest-Subgraph}\xspace }
\newcommand{\DKS}{\textsc{Densest-\(k\)-Subgraph}\xspace }
\newcommand{\BALf}[2]{\textsc{\(#1\)-\(#2\)-Balance}\xspace }
\newcommand{\BAL}{\BALf{\MU}{\NU}\xspace}

\newcommand{\MU}{\mu}
\newcommand{\NU}{\nu}
\newcommand{\BETA}{\beta}
\newcommand{\III}{\mathcal{I}}
\newcommand{\SSS}{\mathcal{S}}
\newcommand{\TTT}{\mathcal{T}}

\newcommand{\SSSet}{\Sigma}

\newcommand{\PPP}{\mathcal{P}}
\newcommand{\RRR}{\mathcal{R}}
\newcommand{\XXX}{\mathcal{X}}
\newcommand{\AAA}{\mathcal{A}}

\let\phi\varphi
\let\eps\epsilon
\newcommand{\argmax}{\mathtt{arg}\max}

\newtheorem{theorem}{Theorem}
\newtheorem{lemma}[theorem]{Lemma}
\newtheorem{corollary}[theorem]{Corollary}
\newtheorem{lemma-rstbl}[theorem]{Lemma}

\makeatletter
\renewcommand\subparagraph{\@startsection{subparagraph}{5}{\z@}%
                                      {1.8ex \@plus1ex \@minus .2ex}%
                                      {-1em}%
                                      {\normalsize\bfseries}}
\makeatother

\title{Balancing spreads of influence in a social network}
\date{Gran Sasso Science Institute (GSSI), L'Aquila, Italy}
\author{Ruben Becker \and Federico Cor\`o \and Gianlorenzo D'Angelo \and Hugo Gilbert}

\begin{document}

\maketitle

\begin{abstract}
    The personalization of our news consumption on social media has a tendency to reinforce our pre-existing beliefs instead of balancing our opinions. This finding is a concern for the health of our democracies which rely on an access to information providing diverse viewpoints.
    To tackle this issue from a computational perspective, Garimella et al.~(NIPS'17) modeled the spread of these viewpoints, also called campaigns, using the well-known independent cascade model introduced by Kempe et al.~(KDD'03) and studied an optimization problem that aims at balancing information exposure in a social network when two opposing campaigns propagate in the network. The objective in their $\mathit{NP}$-hard optimization problem is to maximize the number of people that are exposed to either both or none of the viewpoints. For two different settings, one corresponding to a model where campaigns spread in a correlated manner, and a second one, where the two campaigns spread in a heterogeneous manner, Garimella et al.\ provide constant ratio approximation algorithms.
    In this paper, we investigate a more general formulation of this problem. 
    That is, we assume that $\MU$ different campaigns propagate in a social network and we aim to maximize the number of people that are exposed to either $\NU$ or none of the campaigns, where $\MU\ge\NU\ge2$. We provide dedicated approximation algorithms for both the correlated and heterogeneous settings. Interestingly, while the problem can still be approximated within a constant factor in the correlated setting for any constant values of $\MU$ and $\NU$, for the heterogeneous setting with $\NU\ge 3$, we give reductions leading to several approximation hardness results. Maybe most importantly, we obtain that the problem cannot be approximated within a factor of $n^{-g(n)}$ for any $g(n)=o(1)$ assuming the Gap-ETH hypothesis, denoting with $n$ the number of nodes in the social network.
    For \(\NU \ge 4\), we furthermore show a stronger hardness of approximation bound under a different condition, that is, if a certain class of one-way functions exists, then there is no $n^{-\epsilon}$-approximation algorithm  where $\epsilon > 0$ is a given constant which depends on \(\NU\).
    This complements our finding of an approximation algorithm for the heterogeneous case that for arbitrary $\MU$ and $\NU=3$ leads to an approximation ratio of order $n^{-1/2}$.
\end{abstract}

\section{Introduction}
%
%
%
One of the promises of a highly connected world is that of an impartial spread of opinions driven by free and unbiased sources of information. As a consequence, any opinion could have been equitably exposed to the wide public. On the contrary, the social network platforms that are currently governing news diffusion, while offering many seemingly-desirable features like searching, personalization, and recommendation, are reinforcing the centralization of information spreading and the creation of what is often termed {\em echo chambers} and {\em filter bubbles}~\cite{GarimellaDGM18}. 
Stated differently, algorithmic personalization of news diffusion are likely to create homogeneous polarized clusters where users get less exposure to conflicting viewpoints. A good illustration of this issue was given by Conover et al.\ \cite{ConoverRFGMF11} who studied the Twitter network during the 2010 US congressional midterm elections. The authors demonstrated that the retweet network had a highly segregated partisan structure with extremely limited connectivity between left-wing and right-wing users.
Consequently, instead of giving users a diverse perspective and balancing users opinions by exposing them to challenging ideas, social media platforms are likely to make users more extreme by only exposing them to views that reinforce their pre-existing beliefs \cite{ConoverRFGMF11,DelVicario554}. 

To address this issue from a computational perspective, Garimella et al.\ \cite{DBLP:conf/nips/GarimellaGPT17} introduced the problem of {\em balancing information exposure} in a social network. 
Following the {\em influence maximization paradigm} going back to the seminal work of Kempe et al.\ \cite{Kempe2003maximizing,DBLP:journals/toc/KempeKT15}, their problem involves two opposing viewpoints or campaigns that propagate in a social network following the {\em independent cascade model}. 
Given initial seed sets for both campaigns, a centralized agent is then responsible for selecting a small number of additional seed users for each campaign in order to maximize the number of users that are reached by either both or none of the campaigns. The authors study this problem in two different settings, namely the \emph{heterogeneous} and \emph{correlated} settings. 
The heterogeneous setting corresponds to the general case in which there is no restriction on the probabilities with which the campaigns propagate. Contrarily, in the correlated setting, the probability distributions for different campaigns are identical and completely correlated.
After proving that the optimization problem of balancing information exposure is $\mathit{NP}$-hard, the authors designed efficient approximation algorithms with an approximation ratio of $(1-1/e -\epsilon)/2$ for both settings.
 
\subparagraph*{Our Contribution.} 
We address the main open problem in~\cite{DBLP:conf/nips/GarimellaGPT17}, that is we generalize their optimization problem to a setting with possibly more than two campaigns. More precisely, let $\MU$ and $\NU$ be fixed constants such that $2\le \NU \le \MU$. In our general problem, there are $\MU$ opposing campaigns and the task is to maximize the number of nodes in the network that are reached by at least $\NU$ campaigns or remain oblivious to all of them. We term this problem the $\BALf{\MU}{\NU}$ problem. Interestingly, we obtain results that surprisingly differ from the ones of Garimella et al.\ \cite{DBLP:conf/nips/GarimellaGPT17}. Indeed, while we show that any \BALf{\MU}{\NU} problem can be approximated within a constant factor in the correlated setting (Section~\ref{section: correlated}), we obtain strong approximation hardness results in the heterogeneous setting. In particular, when $\NU\ge 3$, we show that under the Gap Exponential Time Hypothesis \cite{manurangsi2017almost}, there is no $n^{-g(n)}$-approximation algorithm with $g(n)=o(1)$ for the \BALf{\MU}{\NU} problem where $n$ is the number of nodes. Moreover, when \(\NU\ge 4\), we show that if a certain class of one-way functions exists \cite{applebaum2013pseudorandom}, there is no $n^{-\epsilon}$-approximation algorithm for the \BALf{\MU}{\NU} problem where $\epsilon > 0$ is a constant which depends on \(\NU\) (Section~\ref{section: hardness}). We mitigate these hardness results by designing an algorithm with an approximation factor of $\Omega(n^{-1/2})$ for the case where $\NU=3$ and $\MU$ is an arbitrary constant (Section~\ref{section: heterogeneous}).

\subparagraph*{Related work.}
There is a large literature on influence maximization, we refer the interested reader to~\cite{Borgs2014maximizing,DBLP:journals/toc/KempeKT15} and references therein. Here we focus on the literature about multiple campaigns running simultaneously on the same network.
Budak et al.~\cite{budak2011limiting} studied the problem of limiting as much as possible the spread of a ``bad'' campaign by starting the spreading of another ``good'' campaign that blocks the first one. The two campaigns \emph{compete} on the nodes that they reach: once a node becomes active in one campaign it cannot change campaign. They prove that the objective function is monotone and submodular and hence they obtain a constant approximation ratio.
Similar concepts of \emph{competing cascades} in which a node can only participate in one campaign have been studied in several works~\cite{apt2011diffusion,bharathi2007competitive,carnes2007maximizing,Dubey2006competing,kostka2008word,lu2015competition,myers2012clash}. Game theoretic aspects like the existence of Nash equilibria have been also investigated in this case~\cite{alon2010note,goyal2014competitive,tzoumas2012game}.
Borodin et al.~\cite{borodin2017strategyproof} consider the problem of controlling the spread of multiple campaigns by a centralized authority. Each campaign has its own objective function to maximize associated with its spread and the aim of a central authority is to maximize the social welfare defined as the sum of the selfish objective function of each campaign. They propose a truthful mechanism to achieve theoretical guarantees on the social welfare.

Two other works closely related to ours are the ones of 
Aslay et al.~\cite{aslay2018maximizing} and Matakos et al.~\cite{matakos2018tell}. The former work tackles an item-aware information propagation problem in which a centralized agent must recommend some articles to a small set of seed users such that the spread of these articles maximizes the expected diversity of exposure of the agents. The diversity exposure is measured by a sum of agent-dependent functions that takes into account user leanings. The authors show that the $\mathit{NP}$-hard problem they define  amounts to optimizing a monotone and submodular function under a matroid constraint and design a constant factor approximation algorithm. The latter paper models the problem of maximizing the diversity of exposure in a social network as a quadratic knapsack problem. Here also the problem amounts to recommending a set of articles to some users in order to maximize a diversity index taking into account users' leanings and the strength of their connections in the social network. The authors show that the resulting diversity maximization problem is inapproximable and design a polynomial algorithm without an approximation guarantee.

\section{Preliminaries}\label{section: preliminaries}
\subsection{Independent Cascade model}
We introduce the well-known \emph{Independent Cascade model}. We mostly follow the terminology and notation from Kempe et al.~\cite{DBLP:journals/toc/KempeKT15}.
Given a directed graph \(G=(V, E)\), probabilities \(p:E\rightarrow [0,1]\) and an initial node set \(A\subseteq V\) called a set of \emph{seed nodes}. Define \(A_0=A\). For \(t\ge 0\), we call a node \(v\in A_t\) \emph{active} at time \(t\). If a node \(v\) is active at time \(t\ge 0\) but was not active at time \(t-1\), i.e., \(v\in A_t\setminus A_{t-1}\) (formally let \(A_{-1}=\emptyset\)), it tries to activate each neighbor $w$, independently, with a probability of success equal to \(p_{vw}\). In case of success \(w\) becomes active at step \(t+1\), i.e., \(w\in A_{t+1}\).
If at some time \(t^*\ge 0\), we have that \(A_{t^*}=A_{t^*+1}\) we say that the process has \emph{quiesced} and call \(t^*\) the time of quiescence. For an initial set \(A\), we denote with \(\sigma(A) = \E[|A_{t^*}|]\) the expected number of nodes activated at the time of quiescence when running the process with seed nodes \(A\).
Kempe et al.\ showed that this process is equivalent to what is referred to as the \emph{Triggering Model}, see~\cite[Proof of Theorem 4.5]{DBLP:journals/toc/KempeKT15}.
For a node \(v\in V\), let \(N_v\) denote all in-neighbors of \(v\). Here, every node independently picks a \emph{triggering set} \(T_v\subseteq N_v\) according to a distribution over subsets of its in-neighbors, namely \(T_v=S\) with probability \(\prod_{u\in S}p_{uv}\cdot\prod_{u\in N_v\setminus S}(1-p_{uv})\).
For a possible outcome \(X=(T_v)_{v\in V}\) of triggering sets for the nodes \(V\), let \(\rho_X(A)\) be the set of nodes reachable from \(A\) in the outcome \(X\). Note that after sampling \(X\), the quantity \(\rho_X(A)\) is deterministic.
According to Kempe et al.\ \cite{DBLP:journals/toc/KempeKT15}, this model is equivalent to the Independent Cascade model and it holds that \(\sigma(A)=\E_X[|\rho_X(A)|]\), where the expectation is over the outcome profile $X$. While it is not feasible to compute $\rho_X(A)$ for all outcome profiles $X$, it is possible to obtain a \((1\pm\epsilon)\)-approximation to \(\sigma(A)\), with probability at least \(1-\delta\), by sampling \(\Omega(|V|^2\log(1/\delta)/\eps^2)\) possible outcomes \(X\) and computing the average over the corresponding values \(|\rho_X(A)|\), see~\cite[Proposition 4.1]{DBLP:journals/toc/KempeKT15}.


\subsection{The \texorpdfstring{\BALf{\MU}{\NU}}{Balance} problem}\label{section:problem}
Inspired by the work of Garimella et al.~\cite{DBLP:conf/nips/GarimellaGPT17}, we consider several information spread processes, we also call them ``campaigns'', unfolding in parallel, each following the Independent Cascade model described above. Formally, we are given a graph \(G=(V, E)\) and \(\MU\) probability functions \((p_i)_{i\in[\MU]}\), where each $p_i$ is a  probability function as in the Independent Cascade model described above, i.e., \(p_i:E\rightarrow [0,1]\).\footnote{For $n\in \mathbb{N}$, we use $[n]$ to denote the set $\{1,\ldots,n\}$.} For an index \(i\in[\MU]\), let \(X_i=(T_v)_{v\in V}\) be a possible outcome sampled using probabilities \(p_i\). Then for a seed set \(A\subseteq V\), we denote with \(\rho^{(i)}_{X_i}(A)\) the set of nodes reachable from \(A\) in outcome \(X_i\).
For an arbitrary sequence \(\RRR=(R_i)_{i\in[\MU]}\) of subsets of \(V\), we define
\[\textstyle
  \NoSM(\RRR)
  :=\big|(V\setminus \bigcup_{i\in[\MU]} R_i)
  \cup \bigcup_{M\subseteq[\MU]:|M|\ge \NU} \bigcap_{i\in M} R_i\big|
\]
to be the number of nodes that are contained in none or in sufficiently many, i.e., in at least \(\NU\), of the sets in \(\RRR\).
Let \(\XXX=(X_i)_{i\in [\MU]}\) be an outcome profile by letting \(X_i\) be a possible outcome according to distribution \(p_i\). Then, for \(\AAA=(A_i)_{i\in[\MU]}\) with \(A_i\subseteq V\), we denote with
\(
  \rho_\XXX(\AAA)
  =(\rho^{(i)}_{X_i}(A_i))_{i\in[\MU]}
\)
the set of reached nodes in the outcome \(\XXX\) from seed sets \(\AAA\). For two sequences of sets \(\AAA\), \(\AAA'\), and a set \(A\), we let \(\AAA \cup \AAA' = (A_i\cup A'_i)_{i\in[\MU]}\) be the element-wise union and \(\AAA \cap A=(A_i\cap A)_{i\in[\MU]}\) be the element-wise intersection with the set \(A\).

For constant integers \(\MU\ge\NU\ge 2\), we consider the following optimization problem:

\begin{cproblem}{\BAL}
    Input: Graph \(G=(V, E)\), probabilities \(\PPP=(p_i)_{i\in[\MU]}\), seed sets \(\III=(I_i)_{i\in[\MU]}\), and \(k\ge 2\).

    Find: sets \(\SSS=(S_i)_{i\in [\MU]}\) with \(\sum_{i\in[\MU]}|S_i|\le k\), such that \(\Phi^\III_{\MU,\NU}(\SSS)\) is maximum, where
    \(
      \Phi^\III_{\MU,\NU}(\SSS):=\E_\XXX[\NoSM(\rho_\XXX(\III \cup \SSS))].
    \)
\end{cproblem}
We refer to the objective function simply by \(\Phi(\SSS)\), in case \(\III\), \(\MU\), and \(\NU\) are clear from the context. We assume $k\le \NU|V|$ as otherwise the problem becomes trivial by choosing $S_i=V$ for every $i\in[\NU]$. Moreover, we assume w.l.o.g.\ that $|V|\ge \MU$ and $k\geq\NU$, since $|V|$ and $k$ are input parameters and $\MU$ and $\NU$ are constant numbers. Following Garimella et al.~\cite{DBLP:conf/nips/GarimellaGPT17}, we distinguish two settings. (1) The \emph{heterogeneous} setting corresponds to the general case in which there is no restriction on $\PPP$. (2) In the \emph{correlated} setting, the distributions $p_i$ are identical and completely correlated for all $i\in [\MU]$. That is, if an edge $(u, v)$ propagates a campaign to $v$, it propagates all campaigns that reach $u$ to $v$.

\subparagraph{Decomposing the Objective Function.}
In all of our algorithms, we use the approach of decomposing the objective function into summands and approximating the summands separately. For an outcome profile \(\XXX\), and seed sets \(\III=(I_i)_{i\in[\MU]}\), we define \(V^{\ell,\III}_\XXX\subseteq V\), for $\ell=0,\ldots, \MU$, to be the set of nodes that are reached by exactly $\ell$ campaigns \emph{from the seed sets $\III$}.
Formally, for any value $\ell\in [\MU]$,
\[\textstyle
    V_\XXX^{\ell,\III}:=
    \bigcup_{\tau\in \binom{[\MU]}{\ell}}\Big(\bigcap_{i\in \tau} \rho_{X_i}^{(i)}(I_i)
    \setminus
    \bigcup_{j\in [\MU]\setminus\tau} \rho_{X_j}^{(j)}(I_j)\Big),
\]
where $\binom{[\MU]}{\ell}$ denotes the set $\{\tau \subseteq [\MU]~:~|\tau|=\ell\}$.
We write  \(V_\XXX^{\ell}\), if the initial seed sets $\III$ are clear from the context.
In the above definition, by convention an empty union is the empty set, while an empty intersection is the whole universe, here $V$.
Accordingly, we define
\[
    \Phi^\ell(\SSS):=\E_\XXX\Big[\NoSM(\rho_\XXX(\III \cup \SSS)\cap V^{\ell,\III}_\XXX)\Big].
\]
Note that $\Phi^\ell(\SSS)$ is the expected number of nodes that are reached by 
\(0\) or at least \(\NU\) campaigns, resulting from nodes that have  been reached by exactly $\ell$ campaigns from $\III$.
Now, the objective function decomposes as
\[
    \Phi(\SSS)
    =\E_\XXX\Big[\NoSM(\rho_\XXX(\III \cup \SSS))\Big]
    =\E_\XXX\Big[\sum_{\ell\in[\MU]}\NoSM(\rho_\XXX(\III \cup \SSS)\cap V^\ell_\XXX )\Big]
    =\sum_{\ell\in[\MU]}\Phi^{\ell}(\SSS),
\]
using linearity of expectation and that sets $V_\XXX^{\ell}$ are disjoint.
%
Furthermore, we will denote by 
\[
    \Phi^{\ge \ell}(\SSS) := \sum_{i=\ell}^{\mu} \Phi^{i}(\SSS)= \E_\XXX\Big[\NoSM(\rho_\XXX(\III \cup \SSS)\setminus (\cup_{j=0}^{\ell-1} V^j_\XXX ))\Big].
\]
Again, $\Phi^{\ge \ell}(\SSS)$ denotes the expected number of nodes that are reached by sufficiently many campaigns or none of them resulting from nodes that have previously been reached by \emph{at least} \(\ell\) campaigns. Clearly, $\Phi(\SSS)=\Phi^{\ge 0}(\SSS)$. For convenience, in what follows, we will often refer to $\SSS$ as a set of pairs in $\hat{V} := V\times[\MU]$, where picking pair $(v, i)$ into $\SSS$ corresponds to picking $v$ into set $S_i$.
We fix the following observations:
\begin{itemize} 
    \item  For $\ell=0$, $\Phi^0(\SSS)$ is optimal when $\SSS=(\emptyset)_{i\in [\MU]}$. 
    The achieved value is the expected size of $V^0_\XXX$:
    $\Phi^0(\SSS)=\E_\XXX[\NoSM(\rho_\XXX(\III \cup (\emptyset)_{i\in[\MU]})\cap V^0_\XXX)]=\E_\XXX[|V^0_\XXX|]$.
    \item For $\ell=\NU-1$, the function $\Phi^{\ge \NU - 1}(\SSS) = \sum_{ i =\NU -1}^{\MU}\Phi^{i}(\SSS)$ is monotone and submodular.
\end{itemize}

\subparagraph{A First Structural Lemma.}
When applying the standard greedy hill climbing algorithm to finding a set of size $k$ maximizing a submodular set function the key property that is used in the analysis is the following. At any stage of the greedy algorithm there exists an element which leads to an improvement that is at least a fraction of $k$ of the difference of the optimal and the current solution, compare for example~\cite[Lemma 3.13]{Hochbaum:1996:AAN:241938}. Maybe the most important structural lemma underlying our algorithms is a very similar result for the functions $\Phi^{\ge \ell}$.
\begin{lemma}\label{lemma:exists element}
  Let $\ell\in[1,\NU-1]$ and \(\SSS\subseteq \hat V\) with \(|\SSS|\le k-(\NU-\ell)\) and define \(U:=\{\tau \subseteq \hat{V}, |\tau| = \NU-\ell\}\). Then, \(\tau^*=\argmax\{\Phi^{\ge \ell}(S \cup \tau):\tau\in U\}\) satisfies
  \(
     \Phi^{\ge \ell}(\SSS\cup\tau^*) - \Phi^{\ge \ell}(\SSS)
     \ge (\Phi^{\ge \ell}(\SSS^*_{\ge \ell}) - \Phi^{\ge \ell}(\SSS))/\binom{k}{\NU-\ell},
  \)
  where $\SSS^*_{\ge \ell}$ is an optimal solution of size $k$ to maximizing $\Phi^{\ge \ell}$.
\end{lemma}
\begin{proof}
  Let $\XXX$ be an outcome profile and let $v$ be an arbitrary node in $V':=V\setminus \bigcup_{j=0}^{\ell -1}V^j_\XXX$. Let us denote by $\ones^\SSS_\XXX(v)$ the indicator function that is one if $v$ is reached by at least $\NU$ campaigns in outcome profile $\XXX$ from seed sets $\III\cup\SSS$ and zero otherwise. We note that \(\Phi^{\ge \ell}(\SSS)=\E_\XXX\big[\sum_{v\in V'} \ones_\XXX^\SSS(v)\big]\).
  Now, define
  \(
    Y:=\{\tau\subseteq \SSS_{\ge \ell}^{*}:|\tau|=\NU-\ell\}
  \),
  i.e., $Y$ are the sets of nodes in $\SSS^{*}_{\ge \ell}$ of size \(\NU-\ell\).
  We now argue that the following inequality holds for $v$ and $\XXX$:
  \begin{align}\label{formula:for each v}
    \ones_\XXX^{\SSS^*_{\ge \ell}}(v) - \ones_\XXX^{\SSS}(v)
    \leq \sum\nolimits_{\tau\in Y} (\ones_\XXX^{\SSS\cup\tau}(v) - \ones_\XXX^{\SSS}(v)).
  \end{align}
  If the left hand side is not positive, the inequality holds, since the right hand side cannot be negative by monotonicity. Hence, assume that the left hand side is positive. In that case it holds that \(\ones_\XXX^{\SSS^*_{\ge \ell}}(v)=1\), but \(\ones_\XXX^{\SSS}(v)=0\), i.e., in outcome profile $\XXX$, $v$ is reached by at least \(\NU\) campaigns from seed sets $\III\cup\SSS^*_{\ge \ell}$ but not from seed sets $\III\cup\SSS$.
  For such $v$, there must be a set \(\tau\in Y\) such that adding $\tau$ to $\SSS$ results in $v$ being reached by \(\NU\) campaigns (recall that $v\in V'$ and thus $v$ is already reached by at least \(\ell\) campaigns). Thus, there exists a set in \(Y\) that contributes a value of 1 on the right hand side and we may conclude that~\eqref{formula:for each v} holds. Now, using linearity of expectation and~\eqref{formula:for each v}, we obtain
  \begin{align*}
     \Phi^{\ge \ell}(\SSS^*_{\ge \ell}) - \Phi^{\ge \ell}(\SSS)
     &\!=\!\E_\XXX\Big[\sum_{v\in V'} (\ones_\XXX^{\SSS^*_{\ge \ell}}(v) - \ones_\XXX^{\SSS}(v)) \Big]
     \!\le\! \E_\XXX\Big[\sum_{v\in V'} \sum_{\tau\in Y} (\ones_\XXX^{\SSS \cup\tau}(v) - \ones_\XXX^{\SSS}(v)) \Big] .
  \end{align*}
    Using linearity of expectation again, we obtain that the right hand side above is equal to $\sum_{\tau\in Y} (\Phi^{\ge \ell}(\SSS\cup\tau) - \Phi^{\ge \ell}(\SSS))$. Then, the statement follows by the maximality of \(\tau^*\) and the fact that \(|Y|\le\binom{k}{\NU-\ell}\).
\end{proof}

\subparagraph{The Correlated Case.}
For the correlated setting, where probability functions are identical for all campaigns and the cascade processes are completely correlated, we introduce an additional function called $\Psi$. First note that in the correlated setting, the outcome profile $\XXX$ in the definition of $\Phi(\SSS)$ satisfies $X_1=\ldots=X_\MU$. In order to define $\Psi$, we introduce an additional fictitious campaign, call it campaign $0$, that spreads with the same probability $p_0=p_1=\ldots=p_\MU$ as the other $\MU$ campaigns.
We extend the outcome $\XXX=(X_i)_{i\in[\MU]}$ with $X_1=\ldots=X_\MU$ to contain also an identical copy $X_0$ and define $\Psi:2^{V\times\{0\}}\rightarrow[n]$ by
\[
  \Psi(\TTT) := \E_{\XXX}\big[\big\vert\big( \rho_{X_{0}}^{(0)}(\TTT) \cap \bigcup_{j=1}^{\NU-1} V_{\XXX}^{j}\big) \cup \bigcup_{j=\NU}^\MU V_\XXX^j\big\vert\big].
\]
Observe that $\Psi(\TTT)$ measures the expected number of nodes that are either (1) reached by more than $\NU$ campaigns from $\III$ or (2) are reached by at least one campaign from $\III$ and are reached by the fictitious campaign $0$ from $\TTT$. Note that nodes from (1) are already reached by sufficiently many campaigns while nodes from (2) have been reached by some campaign from $\III$ and, as witnessed by $\Psi$, can be reached from the nodes in $\TTT$. 
Note that $\Psi$ is monotone and submodular in $\SSS$ which follows directly from $\sigma$ having these properties.

\subparagraph{Approximating $\Psi$ and $\Phi^{\ge \ell}$.}
As mentioned above, already in the standard independent cascade process, it is not feasible to evaluate the function $\sigma$ exactly. However, $\sigma$ can be approximated to within a factor of $(1\pm \eps)$ by sampling a polynomial number of times.
A very similar approach works for approximating the functions $\Psi$ and $\Phi^{\ge \ell}$ for $\ell\in[0, \NU]$.
That is, there is an algorithm \texttt{approx}\((f, \SSS, \III, \NU, \epsilon, \delta)\) that, for $f\in\{\Psi, \Phi^{\ge 0},\ldots, \Phi^{\ge \NU}\}$, sets $\SSS$ and $\III$, and parameters $\NU, \eps, \delta$
returns a $(1\pm \eps)$-approximation of $f(S)$ with probability $1-\delta$. We prove this fact in Appendix~\ref{appendix:preliminaries} in Lemma~\ref{lemma:approximate f}.
The proof relies on a Chernoff bound and is very similar to the original proof of Proposition~4.1 in~\cite{DBLP:journals/toc/KempeKT15} for the $\sigma$-function.

All of our algorithms are of a greedy flavor, that is, we greedily choose sets in order to build the output set $\SSS$. We investigate the impact of the approximation on this approach in the following lemma.
To this end, let $f$ be a function from $\{\Psi, \Phi^{\ge 1},\ldots, \Phi^{\ge \NU}\}$ and, for some $0<\eps\le 1$, let $\tilde f$ be a $(1\pm\eps')$-approximation of $f$ with $\eps':=\eps/(e\cdot \binom{k}{\lambda(f)})$,
where $\lambda(f)$ depends on $f$, namely $\lambda(f):=\NU-\ell$ for $f=\Phi^{\ge \ell}$ and $\lambda(f):=1$ for $f=\Psi$. We denote with $D_f$ the universe over which $f$ is defined, i.e., $D_f:=\hat V$ for $f=\Phi^{\ge \ell}$, while $D_f:=V\times\{0\}$ for $f=\Psi$.

\begin{restatable}{lemma-rstbl}{approximatedifference}
\label{lemma:approximate difference}
  Let $f$ and $\tilde f$ be as above for some \(0 <\eps\le 1\). Let \(U:=\{\tau \subseteq D_f, |\tau| = \lambda(f)\}\), \(\SSS\subseteq D_f\) with \(|\SSS|\le k-\lambda(f)\), and let \(\SSS^*\) denote a set maximizing \(f\) of size $k$. Then, either
  \[
    f(\SSS)\ge \big(1-\frac{1}{e}\big)\cdot f(\SSS^*)
    \quad\text{ or }\quad
    f(\SSS\cup \tilde{\tau}) - f(\SSS) \ge (1-\eps)\cdot (f(\SSS\cup \tau^*) - f(\SSS)),
  \]
  where
  \(
    \tau^*:=\argmax\{f(S\cup \tau):\tau\in U\}
  \), and
  \(
    \tilde{\tau}:=\argmax\{\tilde f(S\cup \tau):\tau\in U\}
  \).
\end{restatable}
We defer the proof to Appendix~\ref{appendix:preliminaries}.
In summary: either $\SSS$ already yields a \((1-1/e)\)-approximation of the optimum of \(f\) or a set $\tau$ of size $\lambda(f)$ maximizing an approximation \(\tilde f\) of \(f\) can lead to a progress of at least an \((1-\eps)\)-fraction of the maximum progress possible.

\subparagraph{Maximizing \(\Phi^{\ge \NU-1}\) and $\Psi$.}\label{subsec: greedy hill}
Here, we fix the result that the standard greedy hill climbing algorithm, we refer to it as \textsc{Greedy}\((f, \eps, \delta, \III, \NU, k)\), can be applied in order to approximate both \(f\in\{\Phi^{\ge \NU-1},\Psi\}\) to within a factor of $1-1/e-\epsilon$ for any \(0<\eps<1\) with probability at least $1-\delta$ for any $0<\delta\le 1/2$. This is based on the fact that these functions are submodular and monotone set functions. See Appendix~\ref{appendix:preliminaries} for a pseudo-code implementation and a proof of the submodularity property.
Since we can only evaluate \(\Phi^{\ge \NU-1}\) and \(\Psi\) approximately, we obtain the additive $\eps$-term.

\begin{restatable}{lemma-rstbl}{standardgreedy}
  \label{lemma:standard greedy}
  Let $f\in \{\Phi^{\ge \NU-1},\Psi\}$ and let \(0<\eps< 1\) and \(0<\delta\le 1/2\). With probability at least \(1-\delta\), \textsc{Greedy}\((f, \eps, \delta, \III, \NU, k)\) returns \(\SSS\) satisfying
  \(
    f(\SSS) \ge (1 - 1/e - \eps)\cdot f(\SSS^*),
  \)
  where $\SSS^*$ is an optimal solution of size \(k\) to maximizing $f$.
\end{restatable}


\section{Hardness of Approximation for the Heterogeneous Case}\label{section: hardness}
We now let $d\ge 2$ be a constant. In this section, we show that in the heterogeneous setting for \(\NU\ge d+1\), the \BALf{\MU}{\NU} problem is as hard to approximate as the \DKSH{d} problem~\cite{Chlamtac2018densest}. Notably, this result has the following consequences: if \(d=2\) there is no $n^{-g(n)}$-approximation algorithm with $g(n)=o(1)$ for \BALf{\MU}{\NU} under the Gap Exponential Time Hypothesis (Gap-ETH). For general \(d\ge 3\), we get that there is no $n^{-\epsilon}$-approximation algorithm for a given  constant  $\epsilon > 0$ which depends on \(d\) under the assumption that a particular class of one way functions exists~\cite{applebaum2013pseudorandom}.
We recall the definition of the \DKSH{d} problem.

\begin{cproblem}{\DKSH{d}}
        Input: \(d\)-Regular Hypergraph \(G=(V, E)\), integer \(k\ge d\).

        Find: set \(S\subseteq V\) with \(|S|\le k\), s.t.\ \(|E(S)|\) is maximum, where
        \(E(S):=\{e \in E: e\subseteq S \}.\)
\end{cproblem}

A \(d\)-regular hypergraph is a hypergraph in which all hyperedges are composed of exactly \(d\) vertices, where $d$ is a constant. When \(d = 2\), \DKSH{d} is known as the \DKS problem. 
For the hardness of approximation proof, we consider the following transform \(\tau\) of an instance \((G=(V, E), k)\) of the \DKSH{d} problem into an instance \(\tau(G, k)=(\overline G=(\overline V, \overline A), \PPP, \III, \overline k)\) of the \BALf{\MU}{\NU} problem.
\begin{itemize}
    \setlength{\itemsep}{2pt}
    \setlength{\parskip}{0pt}
  \item Define \(\overline V:=V_\boxempty\cup V_\ocircle\), where \(V_\boxempty:=V\), i.e., for each node \(v\in V\), we get a node \(v\) in \(\overline V\). 
  Moreover, let $J:=\binom{[\MU - \NU + d]}{d}$, and \(S_d\) be the set of permutations of \([d]\); we then define \(V_\ocircle\) as
  \(
    V_\ocircle:=\{e^t_{\iota,\pi} : e\in E, \iota \in J, \pi \in S_d, t\in [l]\},
  \)
  i.e., for each edge \(e\in E\), we create $\lambda l$ nodes, where \(l := |V|+1\) 
  and $\lambda :=  |S_d|\cdot|J| = d ! \binom{\MU - \NU + d}{d} $. 
  That is, each set \(\iota\) of \(d\) campaigns in $J$, induces $l$ nodes $e^{t}_{\iota,\pi}, t\in [l]$ for each \(\pi\) in \(S_d\).
  
  \item The arc set \(\overline A\) and the probabilities are defined as
  shown in Figure~\ref{fig:edge} illustrating the case of \(d = 3\) (a more detailed illustration is provided in Appendix \ref{app:hard} in Fig.~\ref{fig:edge-appendix}). We get this scheme in \(\overline G\) for every edge \(e=\{v_{1}, \ldots, v_{d}\} \in E\), for each permutation \(\pi\) in \(S_d\), and for each set  in \( J\)  of \(d\) campaigns.
  
  \item The initial seed sets \(\III\) are defined as $I_1 \!=\! I_2 \!=\! \ldots \!=\! I_{\MU- \NU + d} \!=\! \emptyset$, $I_{\MU - \NU + d+1} \!=\! \ldots \!=\! I_{\MU} \!=\! \overline{V}$. 
  
  \item The budget is the same as in the \DKSH{d} problem, i.e., $\overline{k} = k$.
\end{itemize}

\begin{figure}[t]
  \centering{
    \resizebox{0.9\textwidth}{!}{
      \begin{tikzpicture}[
          scale=.7,
          ->,
          >=stealth',
          shorten >=1pt,
          auto,
          semithick
        ]
        \node[rectangle,draw,text=black, minimum width=1.0cm, minimum height=1.0cm] (U)   at (0,1.7)  {\Large$u$};
        \node[rectangle,draw,text=black, minimum width=1.0cm, minimum height=1.0cm] (V)   at (0,0) {\Large $v$};
        \node[rectangle,draw,text=black, minimum width=1.0cm, minimum height=1.0cm] (W)   at (0,-1.7)  {\Large$w$};

        \node[circle,draw,text=black, minimum size=1.15cm] (B)   at (5,0)  {\large $e^{1}_{\iota,\pi}$};
        \node[circle,draw,text=black, minimum size=1.15cm] (C)   at (10,0) {\large $e^{2}_{\iota,\pi}$};
        \node[circle,draw,text=black, minimum size=1.15cm] (D)   at (15,0)  {\large $e^{3}_{\iota,\pi}$};
        \node[text=black, minimum size=1.15cm] (empty)   at (17.5,0)  {$\ldots$};
        \node[circle,draw,text=black, minimum size=1.15cm] (E)   at (20,0)  {\large $e^{l-1}_{\iota,\pi}$};
        \node[circle,draw,text=black, minimum size=1.15cm] (F)   at (25,0)  {\large $e^{l}_{\iota,\pi}$};

        \path (U) edge [bend left,above ] node {$p_{\pi(i)} = 1$} (B)
        (V) edge [above ] node {$p_{\pi(j)} = 1$} (B)
        (W) edge [bend right,below] node {$p_{\pi(k)} = 1$} (B)
              (B) edge [below=1cm] node[below, label={[align=left] $p_{\pi(i)} = p_{\pi(j)}$\\ $= p_{\pi(k)} = 1$}] {} (C)
              (C) edge [below] node[below, label={[align=left] $p_{\pi(i)} = p_{\pi(j)}$\\ $= p_{\pi(k)} = 1$}] {} (D)
              (E) edge node[below, label={[align=left] $p_{\pi(i)} = p_{\pi(j)}$\\ $= p_{\pi(k)} = 1$}] {} (F);

      \end{tikzpicture}
    }
  }
\caption{This figure illustrates the case \(d=3\). For an hyperedge \(e=\{u,v,w\}\) in \(G\), we get \(d ! \binom{\mu-\NU+d}{d}\) schemes of the above type, one for each set \(\iota=\{i,j,k\}\in J\) and for each way of ordering them given by a permutation \(\pi \in S_d\). Probabilities that are not given are equal to 0.}
\label{fig:edge}
\end{figure}
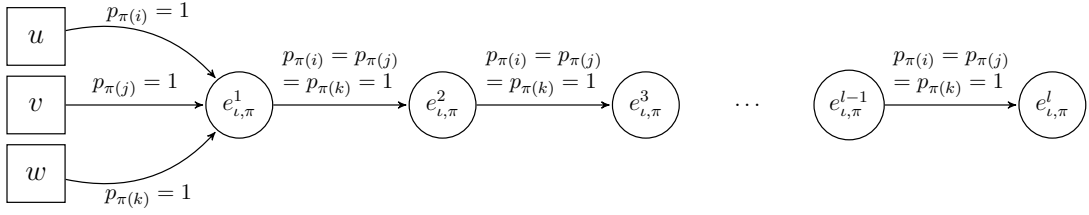

Note that each node in \(\overline G\) is already covered by $\NU-d$ campaigns and that the instance generated is deterministic, in the sense that probability values are either 0 or 1.

Let us now fix a \BALf{\MU}{\NU} instance \(P=(\overline G=(\overline V, \overline A), \PPP, \III, \overline k)\) resulting from the transform \(\tau\) as image of a \DKSH{d} instance \(Q=(G=(V, E), k)\). Clearly, $\overline{V}$ is of cardinality $|V| + \lambda l |E|$ and $\overline{A}$ is of cardinality  $\lambda (l + d-1) |E|$.
Let us denote by $\SSSet$ the set of feasible solutions for \(P\).
For each $\SSS \in \SSSet$, it holds that the objective function $\Phi(\SSS)$ can be decomposed as
\(
  \Phi(\SSS) = \Phi_{\boxempty}(\SSS) + \Phi_{\ocircle}(\SSS),
\)
where 
\[
    \Phi_{\boxempty}(\SSS):=\NoSM(\rho_\XXX(\III \cup \SSS)\cap V_{\boxempty}) \quad\text{ and }\quad \Phi_{\ocircle}(\SSS):=\NoSM(\rho_\XXX(\III \cup \SSS)\cap V_{\ocircle}),
\]
for \(\XXX\) being the only possible (deterministic) outcome profile.
Now, let $\SSS^*$, $\SSS^*_\boxempty$, and $\SSS^*_\ocircle$ denote  optimal solutions to the problem of maximizing $\Phi$, $\Phi_\boxempty$, and $\Phi_\ocircle$, respectively, over $\SSSet$. The following lemma whose proof can be found in Appendix~\ref{app:hard} collects three statements. The first statement says that an optimal solution to $\Phi$ also maximizes \(\Phi_\ocircle\). 
The second statement says that there exists a feasible solution to \(P\) which achieves at least a multiple of \(l \cdot p\) of the objective value in \DKSH{d} with \(p = d!/d^d\). 
In the third statement, we observe that from a feasible solution to \(P\), we can construct a feasible solution to \(Q\) while loosing only a factor of \(\lambda l\) in objective value.

\begin{restatable}{lemma-rstbl}{hardnessclaims}
\label{lemma: hardnessclaims}
    \begin{enumerate}[(1)]
        \item\label{claim:optimality circle} An optimal solution to \(\Phi\) also maximizes \(\Phi_\ocircle\), i.e., \(\Phi_\ocircle(\SSS_\ocircle^*)=\Phi_\ocircle(\SSS^*)\).
        \item\label{claim:relation to lrds} It holds that $\Phi_\ocircle(\SSS^*_\ocircle)\ge l\cdot p \cdot \dksh{d}$, where \(\dksh{d}\) is the optimal value of \DKSH{d} in \(Q\) and \(p = d!/d^d\).
        \item\label{claim:construct S phi} Given $\SSS\in \SSSet$, we can, in polynomial time, build a feasible solution $S$ of \(Q\) such that $|E(S)| \ge \Phi_\ocircle(\SSS)/(\lambda l)$.
    \end{enumerate}
\end{restatable}

We are now ready to show the following relations between the complexity of the two problems. Note that the assumption that $g$ is non-increasing is w.l.o.g.
\begin{theorem}
    Let $d\ge 2$, $\NU\ge d+1$, and \(p = d!/d^d\), then we have the following two cases:
   
          \textbf{Case \(d=2\):} Let $\alpha(n)=n^{-g(n)}$ with $g$ being non-increasing, $g(n)=o(1)$ and \(\alpha(n)\in (0,1]\) and $\beta(n) = \frac{p \cdot  n^{-6g(n)}}{2\lambda}$. 
          
          \textbf{Case \(d \ge 3\):} Let $\alpha(n)=n^{-\eps(d)}$ where \(\eps(d) > 0\) is a constant which depends on \(d\), \(\alpha(n)\in (0,1]\) and  $\beta(n) = \frac{p \cdot n^{-\eps'(d)}}{2\lambda}$, with \(\eps'(d) = (d+4)\cdot \eps(d)\).
    
    In both cases the following statement holds: If there is an $\alpha(|\overline{V}|)$-approximate algorithm for the deterministic \BALf{\MU}{\NU} problem, then there is a $\beta(|V|)$-approximate algorithm for \DKSH{d}. Here $|\overline{V}|$ and $|V|$ denote the number of vertices in the \BALf{\MU}{\NU} and the \DKSH{d} problems, respectively and $\lambda=d ! |J|$.
\end{theorem}
\begin{proof}
    Let \(Q=(G, k)\) be an instance of the \DKSH{d} problem and let \(P:=(\overline G=(\overline V, \overline E), \PPP, \III, \overline k)=\tau(G, k)\) be the instance of the \BALf{\MU}{\NU} problem obtained by the transform \(\tau\).   
    For brevity, let \(n:=|V|\) and \(\overline n:= |\overline V|\). Moreover, let $\SSS$ be an \(\alpha(|\overline V|)\)-approximate solution to \(P\), that is \(\Phi(\SSS)\ge \alpha(|\overline V|) \Phi(\SSS^*)\). We show how to construct a \(\beta(n)\)-approximate solution \(S\) to \(Q\).
    
    Using Lemma~\ref{lemma: hardnessclaims},~(\ref{claim:construct S phi}), we obtain a feasible solution \(S\) to \(Q\) with \(|E(S)| \ge \Phi_\ocircle(\SSS)/(\lambda l)\). We proceed by lower-bounding \(\Phi_\ocircle(\SSS)\). We can w.l.o.g.\ assume that \(\SSS \cap V_{\ocircle}=\emptyset\) and that \(\Phi_\ocircle(\SSS) \ge l\). Indeed, if \(\Phi_\ocircle(\SSS) < l\) then \(\Phi_\ocircle(\SSS) = 0\) and we can build in polynomial-time a better solution by identifying one edge \((v_1,\ldots,v_d)\) and propagating campaign \(i\) in \(v_i\). This further implies that \(\Phi_\ocircle(\SSS)\ge\Phi_\boxempty(\SSS)\) as \(l > n \ge \Phi_\boxempty(\SSS)\). 
   We obtain
    \[
      \Phi_\ocircle(\SSS)
      \ge \frac{\Phi(\SSS)}{2}
      \ge \frac{\alpha(\overline n)}{2}\cdot \Phi(\SSS^*)
      \ge \frac{\alpha(\overline n)}{2}\cdot \Phi_\ocircle(\SSS^*)
      = \frac{\alpha(\overline n)}{2}\cdot \Phi_\ocircle(\SSS^*_\ocircle)
      \ge \frac{\alpha(\overline n)\cdot l\cdot p}{2}\cdot \dksh{d},
    \]
    using Lemma~\ref{lemma: hardnessclaims},~(\ref{claim:optimality circle}) and~(\ref{claim:relation to lrds}) in the last two steps. In summary, we have \(|E(S)| \ge \frac{\alpha(\overline n)\cdot p}{2 \lambda} \dksh{d}\). Note that \(2\lambda/p\) is a constant.
    
        \textbf{Case \(d=2\):} 
        Since $g$ is non-increasing, we get
        \(
            \alpha(\overline n)
            = {\overline n}^{-g(\overline n)}
            = 2^{-g(\overline n)\log(\overline n)}
            \ge 2^{-g(n)\log(2\lambda n^3)}
            \ge 2^{-6 g(n)\log(n)}
            = n^{-6 g(n)},
        \)
        where we used \(2\le n\le \overline n \le 2\lambda n^3\) and $\lambda\le \MU^2\le n^2$ (as \(d=2\)). This completes this case.

        \textbf{Case \(d \ge 3\):}
        In this case
        \(
            \alpha(\overline n)
            = {\overline n}^{-\epsilon(d)}
            = 2^{-\epsilon(d)\log(\overline n)}
            \ge 2^{-\epsilon(d)\log(2\lambda n^3)}
            \ge 2^{-(d+4)\epsilon(d)\log(n)}
            = n^{-(d+4)\epsilon(d)},
        \)
        where we used \(2\le n\le \overline n \le 2\lambda n^3\) and $\lambda\le \MU^d\le n^d$. This completes this case.\qedhere
\end{proof}

To sum up, our reduction shows that: (1) as \DKSH{d} cannot be approximated within $1/n^{\epsilon}$ for some constant \(\epsilon > 0\) which depends on \(d\), if a particular class of one way functions exists  \cite{applebaum2013pseudorandom}, we have shown that the same hardness result holds for any $\BALf{\MU}{\NU}$ problem with $\NU \ge d+1 \ge 4$; (2) moreover as \DKS cannot be approximated within $1/n^{o(1)}$,  if the Gap-ETH holds~\cite{manurangsi2017almost}, we have shown that the same hardness result holds for any $\BALf{\MU}{\NU}$ problem with $\NU \ge 3$. 

Other approximation hardness results exist for \DKS. We review them here, highlighting the hardness results that our reduction implies in each case.
\begin{itemize}
    \item \DKS cannot be approximated within any constant, if the Unique Games with Small Set Expansion (UGSSE) conjecture holds \cite{raghavendra2010graph}. Therefore, under the UGSSE  conjecture it is easy to prove that the reduction given above shows that any $\BALf{\MU}{\NU}$ problem with $\NU\ge 3$ cannot be approximated within any constant.
    \item \DKS cannot be approximated within $n^{-(\log\log n)^{-c}}$, for some constant $c$ if the exponential time hypothesis holds \cite{manurangsi2017almost}. Under the same conjecture, our reduction  implies the same hardness result for any $\BALf{\MU}{\NU}$ problem with $\NU\ge 3$.
\end{itemize}

\section{Approximation Algorithm for the Heterogeneous Case} \label{section: heterogeneous}
Our approach for maximizing $\Phi(\SSS)$ decomposes it as $\Phi(\SSS)=\Phi^0(\SSS) + \Phi^{\ge 1}(\SSS)$ and works on each summand separately. In the following two subsections, we give two different algorithms for maximizing $\Phi^{\ge 1}(\SSS)$. At the end of the section, we show how to combine them.

\subparagraph{Greedily Picking Tuples.}
In this paragraph, we present \textsc{GreedyTuple}\((\eps, \delta, \ell, \III, \NU, k)\) that, for given $\ell$, computes a solution to maximizing $\Phi^{\ge \ell}$. For $\ell=\NU - 1$ the algorithm is identical to the standard greedy hill climbing algorithm. 
For the general case of $\ell\le \NU - 1$, we will show the following theorem. The algorithm is inspired by a greedy algorithm, called Greedy1, due to~\cite{DAngeloetal} for solving the so-called maximum coverage with pairs problem.
\begin{theorem} \label{theorem:greedytuple}
  Let $\epsilon \in (0,1)$, \(\delta\le 1/2\), and $\ell\in [1,\NU - 1]$. If $k\ge 2\NU/\eps$, with probability at least \(1-\delta\), the algorithm \textsc{GreedyTuple}\((\eps, \delta, \ell, \III, \NU, k)\) returns  a solution \(\SSS\) satisfying
  \(
    \Phi^{\ge \ell}(\SSS)
    \ge (1-\frac{1}{e}-\eps)/\binom{k-1}{\NU - \ell - 1}\cdot \Phi^{\ge \ell}(\SSS_{\ge \ell}^*),
  \)
  where $\SSS_{\ge \ell}^*$ is an optimal solution to \(\Phi^{\ge \ell}\).
\end{theorem}

\begin{algorithm}[t]
    \(t:=\lceil \frac{k}{\NU-\ell} \rceil \binom{|\hat V|}{\NU-\ell}\), \(\delta'\leftarrow \delta/t\), \(\eps'\leftarrow \eps/(2e\cdot\binom{k}{\NU-\ell})\), $\SSS \leftarrow \emptyset$\;
    \While{$|\SSS| \leq k - (\NU - \ell)$}{
      Compute $\tau \leftarrow \argmax_{\tau \subseteq \hat{V}, |\tau| = {\nu-\ell}} \{\texttt{approx}(\Phi^{\ge \ell}, \SSS\cup \tau, \III, \NU, \epsilon', \delta')\}$,
      set $\SSS \leftarrow \SSS \cup \tau$
    }
    \Return $\SSS$
\caption{\textsc{GreedyTuple}\((\eps, \delta, \ell, \III, \NU, k)\)}
\label{Greedy1}
\end{algorithm}

We let $\SSS^i$ denote the set $\SSS$ at the end of iteration $i$ of the algorithm. The main idea underlying the analysis of \textsc{GreedyTuple} is very much related to the analysis of the standard greedy algorithm. That is (ignoring the approximation issue), every step of the algorithm incurs a factor of $(1-(1-1/\binom{k}{\NU-\ell})$. For $\ell=\NU-1$, this coincides with the standard case.
\begin{restatable}{lemma}{lemEquiv} \label{lem:equiv1}
    Let $0 < \epsilon <1$, \(\delta\le 1/2\), and $\ell\in [1,\NU - 1]$. With probability at least $1-\delta$, after each iteration $i$ of Algorithm~\ref{Greedy1}, it either holds that
    \[
        \Phi^{\ge \ell}(\SSS^i) \ge \Big(1-\Big(1-\frac{1-\frac{\eps}{2}}{\binom{k}{\NU-\ell}}\Big)^i\Big) \cdot \Phi^{\ge \ell}(\SSS_{\ge \ell}^*) \quad\text{ or }\quad
        \Phi^{\ge \ell}(\SSS^i) \ge \big(1-\frac{1}{e}\big)\cdot \Phi^{\ge \ell}(\SSS_{\ge \ell}^*).
    \]
\end{restatable}
The proof of this lemma can be found in Appendix~\ref{section: appendix greedy tuple}, it uses Lemmata~\ref{lemma:exists element}~and~\ref{lemma:approximate difference}.
We are now ready to give the proof of Theorem~\ref{theorem:greedytuple}.
\begin{proof}[Proof of Theorem~\ref{theorem:greedytuple}]
    Let $\SSS$ denote the set returned by the algorithm. Clearly, $\Phi^{\ge \ell}(\SSS)\ge \Phi^{\ge \ell}(\SSS^\iota)$, where $\iota$ denotes the number of iterations of the while loop in the algorithm.
    By assumption $k\ge 2\NU/\eps$ and thus $\iota = \lfloor \frac{k}{\NU-\ell}\rfloor \ge \frac{k}{\NU-\ell} - 1 \ge (1-\frac{\eps}{2})\cdot\frac{k}{\NU-\ell}$.
    Using Lemma~\ref{lem:equiv1} for $\SSS^\iota$ yields that either $\Phi^{\ge \ell}(\SSS^\iota) \ge (1-1/e)\cdot\Phi^{\ge \ell}(\SSS_{\ge \ell}^*)$ or
    \begin{align*}
        \Phi^{\ge \ell}(\SSS^\iota)
        \ge \Big(1-\Big(1-\frac{1-\frac{\eps}{2}}{\binom{k}{\NU-\ell}}\Big)^\iota\Big) \cdot\Phi^{\ge \ell}(\SSS_{\ge \ell}^*)
        \ge \Big(1-\Big(1-\frac{1-\frac{\eps}{2}}{\binom{k}{\NU-\ell}}\Big)^{(1-\frac{\eps}{2})\frac{k}{\NU-\ell}}\Big)\cdot \Phi^{\ge \ell}(\SSS_{\ge \ell}^*).
    \end{align*}
    For the former case, note that $1-1/e$ is greater than the approximation factor required by the theorem.
    For the latter case note that, as $1-x\leq \exp(-x)$ for any real $x$, we
    have
    \[
      1 - \Big(1-\frac{1-\frac{\eps}{2}}{\binom{k}{\NU-\ell}}\Big)^{(1-\frac{\eps}{2})\frac{k}{\NU-\ell}}
      \ge 1 - \exp\Big(\frac{-(1-\frac{\eps}{2})^2}{\binom{k-1}{\NU-\ell-1}}\Big)
      \ge \frac{1 - \frac{1}{e} - \eps}{\binom{k-1}{\NU-\ell-1}},
    \]
    where the last inequality uses that $1-\exp(-x)\le x\cdot(1-\exp(-1))$ and $(1-1/e)(1-x)\le 1-1/e-x$ for any $x\ge 0$. This completes the proof.
\end{proof}

\subparagraph{Being Iteratively Greedy.}
Recall that, at the beginning of this section, we have defined $\Phi^{\ge \ell}(\SSS):=\E_\XXX[\NoSM(\rho_\XXX(\III \cup \SSS)\setminus (\cup_{j=0}^{\ell-1} V^{j,\III}_\XXX ))]$. We now extend this notation by letting
\[\textstyle
    \Phi^{\ge \ell}_{\BETA}(\RRR,\SSS):=\E_\XXX[\NoSMs_{\MU, \BETA}(\rho_\XXX(\RRR \cup \SSS)\setminus\bigcup_{j=0}^{\ell-1} V^{j,\RRR}_\XXX)]
\]
where \(\ell \in [\NU-1]\) and $\beta\in [\NU]$; we will mainly be working with the case $\beta= \ell +1$.
The function measures the expected number of nodes that are reached by at least $\BETA$ campaigns from $\RRR \cup \SSS$ within the set of nodes that have originally been reached by at least $\ell$ campaigns from $\RRR$.
Our goal now is to maximize \(\Phi(\cdot)\) through the following iterative scheme:
for $\ell$ from $1$ to $\NU-1$, we find sets $\SSS^{[\ell]}$ of size $\lfloor k/(\NU-1)\rfloor$ maximizing \(\Phi^{\ge \ell}_{\ell+1}(\RRR^{[\ell]}, \cdot)\), where $\RRR^{[\ell]}:=\III\cup \bigcup_{j=1}^{\ell-1}\RRR^{[j]}$.
That is, in the $\ell^{th}$ iteration, we maximize the number of nodes reached by $\ell+1$ campaigns that have previously been reached by at least $\ell$ campaigns. The approach is motivated by the observation that, for any $\ell \in [\NU-1]$ and initial sets $\RRR$, the function $\Phi^{\ge \ell}_{\ell+1}(\RRR,\SSS)$ is monotone and submodular in $\SSS$, compare with Section~\ref{subsec: greedy hill} where we used this fact for $\ell=\NU - 1$.
Using Lemma~\ref{lemma:standard greedy} applied to $\Phi^{\ge \ell}_{\ell+1}(\RRR,\cdot)$ with $\NU=\ell+1$ we get that the standard greedy algorithm can be used in order to obtain a $(1-1/e-\eps)$-approximate solution when maximizing $\Phi^{\ge \ell}_{\ell+1}(\RRR,\cdot)$. Note that our algorithm, called \textsc{GreedyIter} is inspired by a similar greedy algorithm called Greedy2 from~\cite{DAngeloetal} that is used there for the maximum coverage with pairs problem. We will prove the following theorem in this section.
\begin{theorem}\label{theorem:greedyiter}
  Let $0<\epsilon <1$ and $\delta \leq 1/2$. With probability $1-\delta$, \textsc{GreedyIter}\((\eps, \delta, \III, \NU, k)\) returns $\SSS$ satisfying
  \(
    \Phi(\SSS)
    \ge \frac{(1-\frac{1}{e}-\eps)^{\NU-1}}{\NU^{2\NU - 3}} (\frac{k}{2|V|})^{\NU-2} \cdot\Phi^{\ge 1}(\III, \SSS^*_{\ge 1}),
  \)
 where $\SSS_{\ge 1}^*$ is a set of cardinality \(k\) maximizing \(\Phi^{\ge 1}(\III,\cdot)\).
\end{theorem}

\begin{algorithm}[t]
  \(\delta'\leftarrow \delta/\NU\), \(\eps'\leftarrow \eps/2\),
  $\RRR^{[1]} \leftarrow \III$\;
  \For{$\ell =1, \ldots, \NU - 1$}{
    $\SSS^{[\ell]} \leftarrow \textsc{Greedy}(\Phi^{\ge \ell}_{\ell+1}(\RRR^{[\ell]}, \cdot),\eps', \delta', \RRR^{[\ell]}, \ell + 1, \lfloor k/(\NU-1) \rfloor)$, $\RRR^{[\ell+1]} \leftarrow \RRR^{[\ell]} \cup \SSS^{[\ell]} $\;
}
\Return $\bigcup_{i=1}^{\NU-1} \SSS^{[i]}$
\caption{\textsc{GreedyIter}\((\eps, \delta, \III, \NU, k)\)}
\label{Greedy2}
\end{algorithm}

The proof of Theorem~\ref{theorem:greedyiter} relies on the following two lemmata whose proofs are given in the appendix, see Section \ref{section: appendix greedy iter}. In a sense the first lemma quantifies the loss in approximation of the first iteration of \textsc{GreedyIter}, while the second lemma quantifies the loss of the later iterations. Both proofs rely on the submodularity of $\Phi^{\ge \ell}_{\ell+1}(\RRR,\cdot)$.
\begin{restatable}{lemma}{lemTwoEquiv}
    \label{lem:equiv2}
  Let $\eps>0$ and assume that $k\ge 2(\NU - 1)/\eps$. If $\SSS^{[1]} \subseteq \hat{V}$ is the set of cardinality $\lfloor k/(\NU-1) \rfloor$ selected in the first iteration of \textsc{GreedyIter}\((\eps, \delta, \III, \NU, k)\), then, with probability at least $1-\delta/\NU$, it holds that
  \(
    \Phi^{\ge 1}_2(\III, \SSS^{[1]})
    \ge\frac{1-\frac{1}{e}-\epsilon}{\NU}\cdot \Phi^{\ge 1}(\III, \SSS^*_{\ge 1}),
  \)
  where $\SSS^*_{\ge 1}$ is a set of cardinality $k$ maximizing \(\Phi^{\ge 1}(\III, \cdot)\).
\end{restatable}

\begin{restatable}{lemma}{lemRecursionGreedyTwo}
    \label{lem:recursion greedy 2}
  Let $\eps>0$, $\ell\ge 2$ and assume that $k\ge 2(\NU - 1)/\eps$. If $\SSS^{[\ell]} \subseteq \hat{V}$ is the set of cardinality $\lfloor k/(\NU-1) \rfloor$ selected in the $\ell$'th iteration of \textsc{GreedyIter}\((\eps, \delta, \III, \NU, k)\), then, with probability at least $1-\delta/\NU$, it holds that
  \(
    \Phi^{\ge \ell}_{\ell+1}(\RRR^{[\ell]}, \SSS^{[\ell]})
    \!\ge\! \frac{(1-\frac{1}{e}-\eps) k}{2(\ell+1)(\NU-1)|V|}\cdot
    \Phi^{\ge\ell - 1}_\ell(\RRR^{[\ell-1]}, \SSS^{[\ell-1]}).
  \)
\end{restatable}

\begin{proof}[Proof of Theorem~\ref{theorem:greedyiter}]
  Since $\SSS=\bigcup_{i=1}^{\NU-1}\SSS^{[i]}$, we obtain $\Phi(\SSS)\ge\Phi^{\ge \NU-1}_{\NU}(\RRR^{[\NU-1]} , \SSS^{\NU - 1})$.
  Using the union bound, $\NU - 2$ times Lemma~\ref{lem:recursion greedy 2} and then Lemma~\ref{lem:equiv2} yield that, with probability at least $1-\delta$, it holds that
  \begin{align*}
    \Phi(\SSS)
    &\ge \Big(\frac{(1-\frac{1}{e}-\eps) k}{2\NU(\NU-1)|V|}\Big)^{\NU-2} \Phi^{\ge 1}_2(\RRR^{[1]}, \SSS^{[1]})
    \ge
    \frac{(1-\frac{1}{e}-\eps)^{\NU-1}}{\NU^{2\NU - 3}} \Big(\frac{k}{2|V|}\Big)^{\NU-2} \Phi^{\ge 1}_\NU(\III, \SSS^*_{\ge 1}).\qedhere
  \end{align*}
\end{proof}

\subparagraph{Algorithm for General Heterogeneous \BALf{\MU}{\NU} problem.}
Our approach to solving the general \BALf{\MU}{\NU} problem is now to use both algorithms presented above. According to Theorem~\ref{theorem:greedytuple}, using \textsc{GreedyTuple}\((\eps, \delta/2, 1, \III, \NU, k)\), we obtain a set $\SSS^1$ that with probability $1-\delta/2$ satisfies
$\Phi^{\ge 1}(\SSS^1) \ge
\alpha_1\cdot \Phi^{\ge 1}(\SSS_{\ge 1}^*)$,
where $\SSS_{\ge 1}^*$ denotes an optimal solution of size $k$ to maximizing $\Phi^{\ge 1}$ and $\alpha_1=(1-\frac{1}{e}-\eps)/\binom{k-1}{\NU - 2}$.
According to Theorem~\ref{theorem:greedyiter}, using \textsc{GreedyIter}\((\eps, \delta/2, \III, \NU, k)\), we obtain a set $\SSS^2$ that with probability $1-\delta/2$ satisfies
$\Phi(\SSS^2)\ge \alpha_2 \cdot\Phi^{\ge 1}(\SSS^*_{\ge 1})$, where $\SSS_{\ge 1}^*$ is as above and $\alpha_2=(1-\frac{1}{e}-\eps)^{\NU-1}(\frac{k}{(2|V|})^{\NU-2}/\NU^{2\NU - 3}$.
Now, we define $\SSS'$ to be the solution that achieves the maximum $\max\{\Phi(\SSS^1),\Phi(\SSS^2)\}$ and $\SSS$ to be the solution that achieves the maximum  $\max\{\Phi(\emptyset),\Phi(\SSS')\}$. We obtain
\(
  2\cdot\Phi(\SSS)
  \ge \Phi(\emptyset) + \Phi(\SSS')
  \ge \Phi^0(\emptyset) + \sqrt{\alpha_1\cdot \alpha_2}\cdot\Phi^{\ge 1}(\SSS_{\ge 1}^*),
\)
using that the maximum $\Phi(\SSS')$ is lower bounded by the geometric mean of $\Phi(\SSS^1)$ and $\Phi(\SSS^2)$, which are in turn lower bounded by $\alpha_1\cdot \Phi^{\ge 1}(\SSS_{\ge 1}^*)$ and $\alpha_2 \cdot\Phi^{\ge 1}(\SSS^*_{\ge 1})$, respectively.
Now, let $\SSS^*$ be an optimal solution of size $k$ to maximizing $\Phi$.
Using that the empty set maximizes $\Phi^0$, we have $\Phi^0(\emptyset)\ge \Phi^0(\SSS^*)$. Furthermore $\Phi^{\ge 1}(\SSS_{\ge 1}^*)\ge \Phi^{\ge 1}(\SSS^*)$, thus
\[
  \Phi(\SSS)
  \ge \frac{\sqrt{\alpha_1\alpha_2}}{2}\cdot (\Phi^0(\SSS^*) + \Phi^{\ge 1}(\SSS^*)))
  =\frac{\sqrt{\alpha_1\alpha_2}}{2}\cdot \Phi(\SSS^*).
\]
Plugging in $\alpha_1$ and $\alpha_2$ and using $k^{\NU - 2}\ge \binom{k-1}{\NU - 2}$, we get the following theorem.
\begin{theorem}
  Let $0<\epsilon <1$ and $\delta \leq 1/2$. There is an algorithm that, with probability $1-\delta$, outputs a solution $\SSS$ that satisfies
  \(
    \Phi(\SSS)
    \ge (1-\frac{1}{e}-\eps)^{\frac{\NU}{2}}(\frac{1}{2|V|})^{\frac{\NU-2}{2}}\NU^{-\frac{2\NU - 3}{2}}\cdot \Phi(\SSS^*),
  \)
  where $\SSS^*$ denotes an optimal solution of size $k$ to maximizing $\Phi(\cdot)$.
\end{theorem}
Note that for $\NU=3$, we obtain an algorithm with an approximation ratio of order \(n^{-1/2}\).

\section{Approximation Algorithm for the Correlated Case}\label{section: correlated}
We now turn to the correlated case. Recall that here the probability functions are identical for all campaigns, i.e., $p_1(e) = \ldots = p_\MU(e)$ for every edge $e \in E$. Moreover, the cascade processes are completely correlated, that is, for any edge $(u,v)$, if node $u$ propagates campaign $i$ to $v$, then node $u$ also propagates all other campaigns that reach it to $v$.

We will consider the same decomposition of the objective function as in the heterogeneous case, i.e., $\Phi(\SSS)$ as $\Phi(\SSS) = \Phi^{\geq 1}(\SSS) + \Phi^0(\SSS)$ for a solution $\SSS$.
Recall that $\Phi^{\geq 1}(\SSS)$ counts the number of nodes that are reached by sufficiently many, i.e.\ $\NU$, campaigns from $\III\cup\SSS$ and have been reached by at least one campaign from $\III$. Similarly, $\Phi^0(\SSS)$ counts nodes that are reached by sufficiently many campaigns or none and have previously been reached by no campaign from $\III$. Clearly, as in the heterogeneous case, $\Phi^0(\SSS)$ is optimal when $\SSS=\emptyset$. 
Differently from the heterogeneous case however, we will see that in the correlated setting, 
there is an approximation algorithm for $\Phi^{\ge 1}$ that achieves a constant factor, namely, $(1-1/e - \eps)/(\NU+1)$. The idea is to pick $\NU$ campaigns and propagate them in the same $\lfloor k/\NU \rfloor$ nodes, exploiting that all campaigns spread in an identical manner.

To that end, we consider the problem of maximizing influence spread with one fictitious campaign, say campaign $0$, spreading with the same probabilities as the others. We will consider the nodes reached by campaign $0$ among the nodes that were (a) reached by at least one campaign from $\III$ and were (b) reached by no more than $\NU$ campaigns from $\III$. For this purpose, we had defined the function $\Psi$ in Section~\ref{section: preliminaries}.
Recall
\[
  \Psi(\TTT) := \E_{\XXX}\Big[\big\vert( \rho_{X_{0}}^{(0)}(\TTT) \cap \bigcup_{j=1}^{\NU-1} V_{\XXX}^{j}) \cup \bigcup_{j=\NU}^\MU V_\XXX^j\big\vert\Big]
\]
and observe that $\Psi(\TTT)$ measures the expected number of nodes that are either (1) reached by more than $\NU$ campaigns from $\III$ or (2) are reached by campaign $0$ from $\TTT$ and were reached by at least one campaign from $\III$.
Recall that we had shown that $\Psi$ is submodular and that the greedy hill-climbing algorithm leads to an approximation factor of at least $1-1/e-\eps$ for any $\eps>0$ when applied to maximizing $\Psi$.
The following lemma whose proof can be found in Appendix~\ref{appendix: correlated} collects three statements. The first statement relates the optimum of $\Psi$ to the optimum of $\Phi^{\ge 1}$. The second statement says that we loose a factor of roughly $\NU$ when choosing a set of size $\lfloor k/\NU\rfloor$ instead of $k$ when maximizing $\Psi$ (this is due to submodularity). The last statement shows that a certain solution $\SSS'$ for $\Phi^{\ge 1}$ constructed from a solution $\TTT$ to $\Psi$ achieves the same value.
\begin{restatable}{lemma-rstbl}{correlatedclaims}
\label{lemma: correlated}
    \begin{enumerate}[(1)]
        \item \label{claim: psi phi} If $\TTT^k\subseteq V\times\{0\}$ is a solution of size $k$ maximizing $\Psi$ and $\SSS^k\subseteq \hat{V}$ is a solution of size $k$ maximizing $\Phi^{\geq 1}$, then $\Psi(\TTT^k) \ge \Phi^{\ge 1}(\SSS^k)$.
        \item \label{claim: k k NU} Let $\eps>0$ and $k\ge \NU/\eps$. If $\TTT^k\subseteq V\times\{0\}$ is a solution of size $k$ maximizing $\Psi$ and $\TTT^{\lfloor k/\NU\rfloor}\subseteq V\times\{0\}$ is a solution of size $\lfloor k/\NU\rfloor$ maximizing $\Psi$, then $\Psi(\TTT^{\lfloor k/\NU\rfloor}) \ge \frac{1-\eps}{\NU+1}\cdot\Psi(\TTT^k)$.
        \item \label{claim: ST} Let $\TTT \subseteq V\times\{0\}$ be of size $\lfloor k/\NU\rfloor$. Then $\SSS':= \{(v,j)| (v,0) \in \TTT, j \in [\NU]\}\subseteq \hat V$ is a set of size at most $k$ such that $\Phi^{\geq 1}(\SSS') = \Psi(\TTT)$.
    \end{enumerate}
\end{restatable}

Now let $0<\eps<1$, $0<\delta\le 1/2$,  $\TTT:=$ \textsc{Greedy}\((\Psi, \eps/2, \delta, \III, \NU, \lfloor k/\NU\rfloor)\), and assume that $k\ge 2\NU/\eps$.
Furthermore, let $\SSS':= \{(v,j)| (v,0) \in \TTT, j \in [\NU]\}\subseteq \hat V$ be as in Lemma~\ref{lemma: correlated},~(\ref{claim: ST}). Then, according to Lemma~\ref{lemma:standard greedy}, it holds that
\(
    \Phi^{\geq 1}(\SSS')
    = \Psi(\TTT)
    \ge \alpha'\cdot\Psi(\TTT^{\lfloor k/\NU\rfloor})
\)
with $\alpha':=1 - 1/e - \eps/2$. Using Lemma~\ref{lemma: correlated},~(\ref{claim: k k NU})~and~(\ref{claim: psi phi}), we obtain
\[
    \Phi^{\geq 1}(\SSS')
    \ge \alpha'\cdot \frac{1-\frac{\eps}{2}}{\NU + 1}\cdot\Psi(\TTT^{k})
    \ge \alpha\cdot\Phi^{\ge 1}(\SSS^*_{\ge 1}),
\]
with $\alpha:=(1-1/e -\eps)/(\NU+1)$ and $\SSS^{*}_{\ge 1}$ being an optimal solution to $\Phi^{\ge 1}$ of size $k$. Now let $\SSS$ be the set among $\SSS'$ and $\emptyset$ that achieves the maximum out of $\Phi(\SSS')$ and $\Phi(\emptyset)$. Then $\SSS$ satisfies
\[
    2\cdot \Phi(\SSS)
    \ge \Phi(\emptyset) + \Phi(\SSS')
    \ge \Phi^0(\SSS^*)+ \alpha\cdot \Phi^{\ge 1}(\SSS^*)
    \ge \alpha \cdot \Phi(\SSS^*),
\]
where $\SSS^*$ is an optimal solution of size $k$ to maximizing $\Phi$. Thus we get the following theorem.
\begin{theorem}\label{theorem:last theorem}
  Let $0<\epsilon <1$ and $\delta \leq 1/2$. In the correlated setting, there is an algorithm that, with probability $1-\delta$, outputs a solution $\SSS$ that satisfies
  \(
    \Phi(\SSS)
    \ge \frac{1-1/e-\eps}{2(\NU + 1)}\cdot \Phi(\SSS^*),
  \)
  where $\SSS^*$ denotes an optimal solution of size $k$ to maximizing $\Phi$.
\end{theorem}

\section{Conclusion and Future Works}\label{section: conclusion}
In this paper, we introduced the $\BALf{\MU}{\NU}$ problem which is the generalization of the problem of balancing information exposure in a social network defined by~\cite{DBLP:conf/nips/GarimellaGPT17}. We studied two settings called the correlated and the heterogeneous setting. While we designed an approximation algorithm with a constant approximation factor in the correlated setting, we obtained an approximation hardness result in the heterogeneous setting stating that it is unlikely to find an $n^{-g(n)}$-approximation algorithm with $g(n)=o(1)$ if \(\NU\ge 3\)  or even a $n^{-\epsilon}$-approximation algorithm where \(\epsilon\) is a constant depending on \(\NU\) if \(\NU\ge 4\). In this setting, we designed an approximation algorithm with approximation ratio $\Omega(n^{-1/2})$ for the case when \(\NU=3\).

Several directions of future work are conceivable. First, it is interesting to improve the approximation guarantee for the $\BALf{\MU}{\NU}$ problem in both settings, most importantly for the heterogeneous case with $\NU>3$. Second, since the $\NU$ parameter in the problem is of a threshold flavor, it would be interesting to investigate a smoother objective function by considering various $\NU$ values, with different weights, such that a node reached by $\NU_1$ campaigns contributes more to the objective function than a node reached by $\NU_2<\NU_1$ campaigns, etc. 

\clearpage
\bibliographystyle{alpha}
\bibliography{bibliography}

\newpage
\appendix
\section{Deferred Proofs for Section~\ref{section: preliminaries}}\label{appendix:preliminaries}
\subparagraph{Approximating $\Phi^{\ge \ell}(\SSS)$ and $\Psi$.}
We use the following algorithm for approximating $f\in\{\Psi, \Phi^{\ge 0},\ldots, \Phi^{\ge \NU}\}$.
\begin{algorithm}[ht]
  \tcp{Note that, if $f=\Psi$, then $\SSS\subseteq V\times \{0\}$, otherwise $\SSS\subseteq\hat V$.}
    $T\leftarrow|V|^2\ln(1/\delta)/\epsilon^2$\;
    \For{$t=1,\ldots, T$}{
        Sample outcome profile $\XXX$\;
        \If{$f=\Psi$}{
          Compute $R\leftarrow \rho_{X_0}^{(0)}(\SSS)$ and
          $n_t\leftarrow|( R \cap \bigcup_{j=1}^{\NU-1} V_{\XXX}^{j}) \cup \bigcup_{j=\NU}^\MU V_\XXX^j|$\;
        }\Else{
          Compute $\RRR\leftarrow (\rho_{\XXX_i}^{(i)}(I_i\cup S_i))_{i\in[\MU]}$ and
          $n_t\leftarrow\NoSM(\RRR\setminus (\cup_{j=0}^{\ell -1} V^j_\XXX))$\;
        }
    }
    \Return{$\frac{1}{T}\sum_{t=1}^Tn_t$}
\caption{\texttt{approx}\((f, \SSS, \III, \NU, \epsilon, \delta)\)}
\label{algorithm:evaluate phi_ell}
\end{algorithm}

We show the following lemma. As a condition for the lemma, we have the requirement that $f(\SSS)\ge 1$ for the set $\SSS$ that we evaluate $f$ on. We argue at the end of this section, see Lemma~\ref{lemma:lower bound} that we can assume $\Phi^{\ge \ell}(\SSS)\ge 1$ for any $\ell\in[0,\NU]$ and $\SSS$ as well as $\Psi(\TTT)\ge 1$ for any $\TTT$ at the cost of an arbitrarily small $\eps$ in the approximation guarantee.

\begin{lemma}
\label{lemma:approximate f}
  Let $f\in\{\Psi, \Phi^{\ge 0},\ldots, \Phi^{\ge \NU}\}$ and let \(\SSS\) be such that $f(\SSS)\ge 1$. Let $\tilde f(\SSS):=\normalfont{\texttt{approx}}(f, \SSS, \III, \NU, \epsilon, \delta)$ for some $0<\delta\le 1/2$ and $0<\epsilon<1$,
  then $\tilde f(\SSS)$ is a $(1\pm \epsilon)$-approximation of $f(\SSS)$ with probability at least $1-\delta$.
\end{lemma}
\begin{proof}
  The proof is very similar to the proof of Proposition~4.1 in~\cite{DBLP:journals/toc/KempeKT15}, it is a straightforward application of a Chernoff bound, we use Theorem 2.3 from~\cite{McDiarmid1998} here.
  Let us define $T$ random variables, one for each iteration of the algorithm, $Y_1,\ldots, Y_T$ by $Y_t:=n_t/|V|$. Note that the $Y_t$ are independent and $Y_t\in[0,1]$. Let $S_T:=\sum_{t=1}^TY_t$ and $\mu:=\E[S_T]$, then $S_T=T\cdot \tilde f(\SSS)/|V|$ and $\mu=T\cdot f(\SSS)/|V|$. Thus, setting $\gamma:=\epsilon f(\SSS)/|V|$, the Chernoff bound yields
   \[
    \Pr[|f(\SSS)-\tilde f(\SSS)|\ge \epsilon f(\SSS)]
    =\Pr[|S_T-\mu|\ge T\gamma]
    \le 2 e^{-2T\gamma^2}
    = 2 e^{-\frac{2T\epsilon^2 f(\SSS)^2}{|V|^2}}
    \le \delta,
   \]
   since $T=|V|^2\ln(1/\delta)/\epsilon^2$ and $f(\SSS)\ge 1$.
\end{proof}

Motivated by Lemma \ref{lemma:exists element}, we now investigate how the error in approximating \(f\in\{\Psi, \Phi^{\ge 1},\) \(\ldots, \Phi^{\ge \NU}\}\) affects the error of the difference \(f(S\cup v) - f(S)\).
In other words, we quantify how much we loose while maximizing \(f\) by picking an element \(\tau\) with respect to an approximation of \(f\) only: For some $0<\eps\le 1$, let $\tilde f$ be a $(1\pm\eps')$-approximation of $f$ with $\eps':=\eps/(e\cdot \binom{k}{\lambda(f)})$,
where $\lambda(f)$ is a constant that depends on $f$, namely $\lambda(f):=\NU-\ell$ for $f=\Phi^{\ge \ell}$ and $\lambda(f):=1$ for $f=\Psi$. We get the following lemma.

\approximatedifference*
\begin{proof}
  We distinguish two cases. First, assume that \(f(\SSS\cup \tau^*) - f(\SSS)\le f(\SSS^*)/(e\cdot\binom{k}{\NU-\ell})\).
  If $f=\Phi^{\ell}$ for some $\ell$, Lemma~\ref{lemma:exists element} yields that $\tau^*$ satisfies
  \begin{align}\label{formula:taustar}
    f(\SSS\cup \tau^*) - f(\SSS)\ge \frac{1}{\binom{k}{\lambda(f)}}\cdot (f(\SSS^*)-f(\SSS)).
  \end{align}
  For $f=\Psi$, we note that $\lambda(f)=1$ and thus $\binom{k}{\lambda(f)}=k$, so we get inequality~\eqref{formula:taustar} by the submodularity of $\Psi$.
  Thus, in both case by combining the two inequalities, we get \(f(\SSS)\ge(1 - 1/e)\cdot f(\SSS^*)\), which concludes this case.

 On the other hand, assume \(f(\SSS\cup \tau^*) - f(\SSS)> f(\SSS^*)/(e\cdot\binom{k}{\lambda(f)})\).
  Using the approximation guarantee of \(\tilde f\), the definition of \(\tilde{\tau}\), and again the approximation guarantee, we get
  \begin{align*}
    f(\SSS\cup &\tilde{\tau}) - f(\SSS)
    \ge \frac{1-\eps'}{1+\eps'}f(\SSS\cup \tau^*) - f(\SSS)\\
    &= f(\SSS\cup \tau^*) - f(\SSS) - \frac{2\eps'}{1+\eps'} f(\SSS\cup \tau^*)\\
    &=(1-\eps)\cdot(f(\SSS\cup \tau^*) - f(\SSS)) +\frac{\eps(1+\eps') - 2\eps'}{1+\eps'}\cdot f(\SSS\cup \tau^*) - \eps \cdot f(\SSS).
  \end{align*}
  Thus, it remains to argue that the latter two summands are non-negative. From the case assumption and the optimality of \(\SSS^*\), we have \(f(\SSS) \le f(\SSS\cup \tau^*)/(1 + 1/(e\cdot\binom{k}{\lambda(f)}))\) and thus the above latter two summands can be lower bounded by
  \begin{align*}
    f(S\cup \tau^*) \cdot \Big(\frac{\eps(1+\eps') - 2\eps'}{1+\eps'} - \frac{\eps}{(1+ 1/(e\cdot\binom{k}{\lambda(f)})} \Big).
  \end{align*}
  The latter is non-negative, since
  \[
    (\eps(1+\eps')-2\eps')\Big(1+\frac{1}{e\cdot\binom{k}{\lambda(f)}}\Big)
    - \eps(1+ \eps')
    = \eps'(1+\eps')
    - \frac{2\eps'}{e\cdot\binom{k}{\lambda(f)}}
    \ge 0
  \]
  by the choice of \(\eps':=\eps/(e\cdot\binom{k}{\lambda(f)})\) and \(0 <\eps\le 1\). This concludes the proof.
\end{proof}

\subparagraph{Lower Bound on $\Phi^{\ge \ell}$ and $\Psi$.}
Our goal in this section is to argue that there is a transform $\tau$ that takes an instance $(G, \PPP, \III, k)$ of the $\BALf{\MU}{\NU}$ problem and outputs a (slightly modified) instance $(G', \PPP, \III', k):=\tau(G, \PPP, \III, k)$ such that the function $f'(\SSS)$ is at least $1$ for every argument $\SSS$ in the new instance for any $f\in \{\Phi^{\ge 0}, \ldots, \Phi^{\ge \NU},\Psi\}$.
Moreover, given an approximation algorithm for $\Phi$ with approximation ratio $\alpha$, we will show that applying this algorithm on the transformed instance $\tau(G, \PPP, \III, k)$ leads to a solution of approximation ratio at least $\alpha-\eps$ for the original instance, for any $\eps>0$.

The transform $\tau$ is defined as follows. Obtain $G'$ by adding an isolated node $v$ to $G$ and extend $\III$ to $\III'$ by adding $v$ to $I_i$ for every $i\in[\NU]$. Now clearly, for every solution $\SSS$, it holds that $f'(\SSS)=f(\SSS) + 1\ge 1$, where the 1 originates from the additional node $v$ that is initially covered by $\NU\ge \ell$ campaigns. Moreover, we get the following lemma.
\begin{lemma}\label{lemma:lower bound}
    Let $\eps>0$. Then, for instances $(G, \PPP, \III, k)$ with $k\ge 2\NU/\eps$, the following holds:
    Let $\SSS'$ be a solution in $(G', \PPP, \III', k):=\tau(G, \PPP, \III, k)$ such that $\Phi'(\SSS')\ge \alpha\cdot\Phi'(\SSS^{\prime*})$, where $\SSS^{\prime*}$ denotes an optimal solution for maximizing $\Phi'$ in the new instance $(G', \PPP, \III', k)$.
    Then $\SSS:=\SSS'\setminus \{v\}$ satisfies $\Phi(\SSS)\ge (\alpha-\eps)\cdot\Phi(\SSS^{*})$, where $\SSS^{*}$ denotes an optimal solution for maximizing $\Phi$ in the original instance $(G, \PPP, \III, k)$.
\end{lemma}
\begin{proof}
    First note that $\Phi(\SSS^{*})\ge \lfloor k/\NU\rfloor\ge k/\NU -1\ge 1/\eps$ or equivalently $1\le \eps \Phi(\SSS^{*})$. This yields the claim, since
    \[
        \Phi(\SSS)
        = \Phi'(\SSS) - 1
        \ge \alpha\cdot\Phi'(\SSS^{\prime*}) - 1
        \ge \alpha\cdot\Phi(\SSS^{*}) - 1
        \ge (\alpha - \eps)\cdot \Phi^{\ge\ell}(\SSS^{*}).\qedhere
    \]
\end{proof}
\subparagraph{Maximizing $\Phi^{\ge \NU -1}$ and $\Psi$.}
Our goal here is to show that the standard greedy hill climbing algorithm, we refer to it as \textsc{Greedy}\((f, \eps, \delta, \III, \NU, k)\), can be applied in order to approximate both \(\Phi^{\ge \NU-1}\) and \(\Psi\) to within a factor of $1-1/e-\epsilon$ for any \(0<\eps<1\) with probability at least $1-\delta$ for any $0<\delta\le 1/2$.
We first formally prove that these functions are submodular.

\begin{lemma}
    The functions $\Psi$ and $\Phi^{\ge\NU-1}$ are monotone and submodular.
\end{lemma}
\begin{proof}
    The monotonicity of $\Psi$ and $\Phi^{\ge \NU-1}$ is straightforward. We argue the submodularity of $\Psi$ ($\Phi^{\ge \NU-1}$) in a similar way as we argued in the proof of Lemma~\ref{lemma:exists element}. To this end, let $D(\Psi)=V\times \{0\}$ and $D(\Phi^{\ge \NU - 1})=\hat V=V\times [\MU]$ denote the domain of $\Psi$ and $\Phi^{\ge \NU-1}$, respectively, and let \(\SSS\) and \(\SSS'\) be subsets of $D(\Psi)$ ($D(\Phi^{\ge \NU-1})$) such that \(\SSS \subseteq \SSS'\), and let \(\tau\) be an element of the domain \(D(\Psi)\) ($D(\Phi^{\ge \NU-1})$). Furthermore, let \(\XXX\) be an outcome profile w.r.t.~the correlated (heterogeneous) probability distributions. Lastly, let \(v\in  V\setminus V_\XXX^0\) (\(v\in  V\setminus \bigcup_{j=0}^{\NU-2}V_\XXX^j\)) be a node that can contribute to the value of $\Psi$ ($\Phi^{\ge \NU-1}$).\footnote{Recall that $V_\XXX^j$ is the set of nodes that was reached by $j$ campaigns from seed sets $\III$.} We denote by \(\ones_\XXX^\SSS(v)\) the indicator function that is 1 if \(v\) contributes to \(\Psi\) ($\Phi^{\ge \NU-1}$) in outcome profile \(\XXX\) with initial seed sets \(\III\) and additional seed sets \(\SSS\) and 0 otherwise. We now argue that the following inequality holds:.
    \begin{align}\label{formula: ones}
      \ones_\XXX^{\SSS'\cup\tau}(v) - \ones_\XXX^{\SSS'}(v) \leq \ones_\XXX^{\SSS\cup\tau}(v) - \ones_\XXX^\SSS(v)
    \end{align}
    Note that the right-hand side cannot be negative by monotonicity and that, if the left-hand side is positive for \(\Psi\) (\(\Phi^{\ge \NU-1}\)), then it must hold that the node \(v\) is reached by a subset \(M\subseteq[\MU]\) of campaigns from \(\III\) with \(|M| \in [1,\NU -1]\) (\(|M|=\NU - 1\)). Furthermore, the node \(v\) is not reached by campaign $0$ (is not reached by a campaign \(j\in [\MU]\setminus M\)) from \(\SSS'\), but it is reached by campaign $0$ (campaign $j$) from \(\tau\). Now, observe that \(\SSS \subseteq \SSS'\) and thus the node \(v\) is not reached by campaign 0 (by campaign $j$) from \(\SSS\) neither. Hence it follows that the right-hand side is also positive. Taking the expected value on both sides of~\eqref{formula: ones} yields that $\Psi(\SSS'\cup\tau)-\Psi(\SSS')\le \Psi(\SSS\cup\tau)-\Psi(\SSS)$ ($\Phi^{\ge \NU-1}(\SSS'\cup\tau)-\Phi^{\ge \NU-1}(\SSS')\le \Phi^{\ge \NU-1}(\SSS\cup\tau)-\Phi^{\ge \NU-1}(\SSS))$ due to linearity of expectation. This establishes submodularity and concludes the proof.
\end{proof}

We now recall the following classical result concerning the greedy algorithm for maximizing a submodular function:
\begin{lemma}[Theorem 3.9 in~\cite{Hochbaum:1996:AAN:241938}]\label{theorem:hochbaum}
 The greedy hill-climbing algorithm, that at each step picks an element that leads to an increment being within factor $\beta$ of the optimal increment possible, achieves an approximation ratio of at least \(1-(1-\beta/k)^k>1-1/e^\beta\).
\end{lemma}

We have seen that both \(\Phi^{\ge\NU-1}\) and $\Psi$ can be approximated within a $(1\pm\eps)$-factor using the \texttt{approx}-routine. In Lemma~\ref{lemma:approximate difference} we argued that using the approximations we can find an element $v$ (or a set $\tau$ of cardinality $\lambda(f)=1$) that when added to $\SSS$ leads to a progress of at least a factor of $(1-\eps)$ of the maximal progress possible. We prove Lemma~\ref{lemma:standard greedy}.

\begin{algorithm}[ht]
    \(\delta'\leftarrow \delta/(k|\hat V|)\), \(\eps'\leftarrow \eps/(ek)\), $\SSS\leftarrow \emptyset$\;
    \While{\(|\SSS|\le k\)}{
        Compute \(v \leftarrow \argmax\{\texttt{approx}(f, \SSS \cup \{v\}, \III, \NU, \eps', \delta'):v \in D_f\}\), set \(\SSS\leftarrow \SSS\cup\{v\}\)\;
    }
    \Return{\(\SSS\)}
\caption{\textsc{Greedy}\((f, \eps, \delta, \III, \NU, k)\)}
\label{algorithm:greedy hill}
\end{algorithm}

\standardgreedy*
\begin{proof}
  The union bound over all at most $k|\hat V|$ calls to \texttt{approx}, yields that,
  with probability at least \(1-\delta\), each call resulted in a \(1\pm\eps'\)-approximation.
  Then Lemma~\ref{lemma:approximate difference} applied to \(f\) guarantees that after each iteration $i$ either an element \(v\) is picked such that the increment using \(v\) is at least a \((1-\eps)\)-fraction of the optimal increment possible in this iteration or the current set $\SSS^i$ is already a \((1-1/e)\)-approximation of the optimum set $\SSS^*_{\NU-1}$. In the latter case the lemma is fulfilled by the monotonicity of \(f\).
  In the former case we get an \(\SSS\) having an approximation ratio of at least
  $1 - (1-(1-\eps)/k)^{k}
    \ge 1 - 1/e^{1-\eps}$ according to Lemma~\ref{theorem:hochbaum}.
  Since
  \(
    1 - \frac{1}{e^{1-\eps}}
    \ge (1-\eps)\cdot \big(1-\frac{1}{e}\big)
    \ge 1 - \frac{1}{e} - \eps ,
  \)
  this concludes the proof.
\end{proof}

\section{Deferred Proofs for Section~\ref{section: hardness}} \label{app:hard}

\paragraph{Further illustration of the reduction described in  Section \ref{section: hardness}.}
\begin{figure}[t]
  \centering{
    \resizebox{0.99\textwidth}{!}{
      \begin{tikzpicture}[
          scale=.4,
          ->,
          >=stealth',
          shorten >=1pt,
          auto,
          semithick,
          rubberduck/.style={
            draw=red!50,
            shape=isosceles triangle,
            fill=red!50,
            minimum height=6.0cm,
            minimum width=2.5cm,
            shape border rotate=#1,
            isosceles triangle stretches,
            inner sep=1pt,
          },
          rubber/.style={rubberduck=0},
          ducky/.style={rubberduck=180},
          every node/.style={sloped,anchor=south,auto=false}
        ]

        \node[ducky] (B1) at (25,11) {\Large $P_{[3], (2,3,1)}$};
        \node[ducky] (B2) at (25,0) {\Large $P_{[3], (3,1,2)}$};
        \node[ducky] (B3) at (25,-11) {\Large $P_{[3], (3,2,1)}$};

        \node[rubber] (B4) at (-25,11) {\Large $P_{[3], (1,2,3)}$};
        \node[rubber] (B5) at (-25,0) {\Large $P_{[3], (1,3,2)}$};
        \node[rubber] (B6) at (-25,-11) {\Large $P_{[3], (2,1,3)}$};

        \node[rectangle,draw,text=black, minimum width=1.0cm, minimum height=1.0cm] (U)   at (0,4)  {\Large$u$};
        \node[rectangle,draw,text=black, minimum width=1.0cm, minimum height=1.0cm] (V)   at (0,0) {\Large $v$};
        \node[rectangle,draw,text=black, minimum width=1.0cm, minimum height=1.0cm] (W)   at (0,-4)  {\Large$w$};

        \path (U) edge node {$p_2 = 1$} (B1.west)
        (V) edge [above left] node {$p_3 = 1$} (B1.west)
        (W) edge [below right=1mm] node {$p_1 = 1$} (B1.west);
        \path (U) edge [above ] node {$p_3 = 1$} (B2.west)
        (V) edge [above ] node {$p_1 = 1$} (B2.west)
        (W) edge [below] node {$p_2 = 1$} (B2.west);
        \path (U) edge [above ] node {$p_3 = 1$} (B3.west)
        (V) edge [above ] node {$p_2 = 1$} (B3.west)
        (W) edge [below] node {$p_1 = 1$} (B3.west);
        \path (U) edge [above ] node {$p_1 = 1$} (B4.east)
        (V) edge [above ] node {$p_2 = 1$} (B4.east)
        (W) edge [below] node {$p_3 = 1$} (B4.east);
        \path (U) edge [above ] node {$p_1 = 1$} (B5.east)
        (V) edge [above ] node {$p_3 = 1$} (B5.east)
        (W) edge [below] node {$p_2 = 1$} (B5.east);
        \path (U) edge [above ] node {$p_2 = 1$} (B6.east)
        (V) edge [above ] node {$p_1 = 1$} (B6.east)
        (W) edge [below] node {$p_3 = 1$} (B6.east);
      \end{tikzpicture}
    }
  }
\caption{For a set $\iota=\{i,j,k\}\in J$ of $d$ campaigns and a permutation $\pi\in S_d$,
let $P_{\iota, \pi}$ stand for the path in Figure~\ref{fig:edge} of nodes $e_{\iota,\pi}^1,\ldots,e_{\iota,\pi}^l$ connected by arcs $(e_{\iota,\pi}^t, e_{\iota,\pi}^{t+1})$ for $t=1,\ldots, l-1$ with probabilities on these edges being one for $p_{\pi(i)}$, $p_{\pi(j)}$, and $p_{\pi(k)}$ and zero for all other indices.
The figure illustrates the case $d=3$ and $\MU=\NU=4$ and the portion of the network that is generated  in the transform \(\tau\) of Section\ref{section: hardness} for one hyper-edge $e=\{u,v,w\}$ and the only set $\iota=[3]\in J = \binom{3}{3}$. Probabilities that are not given are equal to zero.}
\label{fig:edge-appendix}
\end{figure}
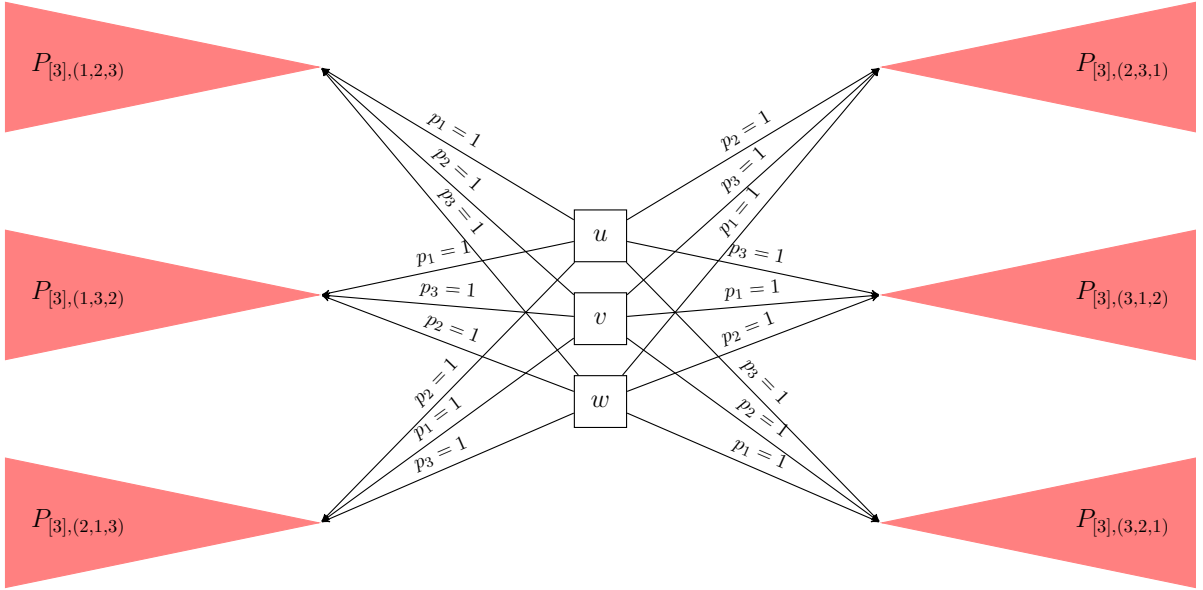
Figure \ref{fig:edge-appendix} illustrates the scheme induced by an hyperedge \(e = (u,v,w)\) when \(d = 3\) and \(\MU = \NU = 4\). In this case, $J=\binom{[\MU - \NU + d]}{d}$ is only composed of set \([3]\) and \(S_3\) is composed of 6 permutations. We use the standard tuple notation for permutations.
\paragraph{Deferred Proofs for Reducing \texorpdfstring{\DKSH{d}}{DKSH(d)} to \texorpdfstring{\BALf{\MU}{\NU}}{Balance}.}
We start by defining the \MCDSH{d} problem which is closely related to the \DKSH{d} problem.
\begin{cproblem}{\MCDSH{d}}
        Input: \(d\)-Regular Hypergraph \(G=(V, E)\), integer \(k\)

        Find: set \(S\subseteq V\) with \(|S| \le k\) and a coloring function \(\phi : S \rightarrow [d]\), s.t. \(|E_\phi(S)|\) is maximal, where
        \(
          E_\phi(S) := \{e \in E : e = (v_{1},\ldots,v_{d}) \subseteq S ~\wedge~ \phi(v_{i}) \neq \phi(v_{j}), \forall i\neq j\}.
        \)
\end{cproblem}

Problem \MCDSH{d} will be of interest to us due to the following results.
We first prove a lemma showing the existence of an assignment $\phi'$ such that at least a fraction \(p\) of the hyperedges in the induced sub-hypergraph of a set \(S\) have differently colored endpoints.

\begin{lemma}
\label{lemma:greedy assignment}
  Let \(G=(V, E)\) be a \(d\)-regular hypergraph. For any set \(S\subseteq V\), there is an assignment \(\phi':S\rightarrow [d]\) s.t. \(|E_{\phi'}(S)|\ge p\cdot |E(S)|\) where \(p = d! / d^{d}\).
\end{lemma}
\begin{proof}
    Let \(S\subseteq V\). Consider the probabilistic procedure in which, for each node, we assign a color from $[d]$ uniformly at random and independently of the other nodes. This procedures yields a coloring \(\phi\). For any \(e = (v_{1},\ldots,v_{d})\) in \(S\), the probability that \(\phi(v_{i})\neq \phi(v_{j})\) for all \(i\neq j\) is \(p\). This property is guaranteed if and only if \( (\phi(v_{1}),\ldots,\phi(v_{d}))\) corresponds to one of the \(d!\) permutation of \([d]\). In total, there are \(d^{d}\) ways of coloring \(e\). Hence, the expected value of \(|E_\phi(S)|\) is \(p\cdot |E(S)|\). Consequently, the function \(\phi'\) that maximizes \(|E_{\phi'}(S)|\) satisfies \(|E_{\phi'}(S)|\geq p\cdot |E(S)|\).
\end{proof}

This leads to the following corollary.
\begin{corollary}\label{corollary:2opts}
    Denoting with \(\dksh{d}(G, k)\) and \(\mcdsh{d}(G, k)\) the value of the optimal solution for \DKSH{d} and \MCDSH{d} on \((G, k)\), respectively, we have that \(\dksh{d}(G, k)\le \mcdsh{p}(G,k)/p\), where \(p = d! / d^{d}\).
\end{corollary}
\begin{proof}
  Let \(S\) be a set that achieves \(\dksh{d}(G,k)=|E(S)|\), then \(\dksh{d}(G,k)=|E(S)|\le |E_{\phi'}(S)|/p\le \mcdsh{d}(G,k)/p\), where \(\phi'\) is as in Lemma~\ref{lemma:greedy assignment}.
\end{proof}

Recall that we fixed a \BALf{\MU}{\NU} instance \(P=(\overline G=(\overline V, \overline A), \PPP, \III, \overline k)\) resulting from the transform \(\tau\) as image of an \DKSH{d} instance \(Q=(G=(V, E), k)\). 
In what follows nodes in \(V_\boxempty\) (resp. \(V_\ocircle\))  are called rectangle-nodes (resp. circle-nodes).
\hardnessclaims*
\begin{proof}
  \begin{enumerate}
        \item We can w.l.o.g.\ assume that \(\SSS^*\cap V_{\ocircle}=\emptyset\) and \(\SSS_\ocircle^*\cap V_{\ocircle}=\emptyset\). Then, it follows that both \(\Phi_\ocircle(\SSS^*)\) and \(\Phi_\ocircle(\SSS_\ocircle^*)\) are multiples of \(l\). Now, assume for the purpose of contradiction that \(\Phi_\ocircle(\SSS_\ocircle^*)>\Phi_\ocircle(\SSS^*)\). Then, \(\Phi_\ocircle(\SSS^*_\ocircle) \ge \Phi_\ocircle(\SSS^*) + l\) which leads to
        \[
            \Phi(\SSS^*)=\Phi_\ocircle(\SSS^*) + \Phi_\boxempty(\SSS^*) \le \Phi_\ocircle(\SSS^*_\ocircle) - l + |V| < \Phi(\SSS^*_\ocircle),
        \]
        using that \(l> |V|\). This is a contradiction to \(\SSS^*\) being optimal.
        \item Let $(S^*,\phi^*)$ be an optimal solution to the \MCDSH{d} problem induced by \(Q\). Construct a solution $\SSS$ for \BALf{\MU}{\NU} by letting \(S_i:=\{v\in V:\phi(v)=i\}, \forall i \in [d]\). Clearly \(\Phi_\ocircle(\SSS)= l |E_{\phi^*}(S^*)| \). Thus, using Corollary \ref{corollary:2opts}:
        \(
        \Phi_\ocircle(\SSS^*_\ocircle)\ge l \cdot \mcdsh{d} \geq l \cdot p\cdot \dksh{d}.
        \)
        \item  Let $S=\{v\in V_\boxempty:v\in S_i \text{ for some }i\in[\MU-\NU+d]\} \subseteq V_\boxempty=V$ be the set of rectangle-nodes where $\SSS$ propagates at least one campaign in $[\MU-\NU+d]$. Clearly, $|S| \le k$. Let $q=|E(S)|$ be the number of edges in the sub-graph of \(G\) induced by \(S\). Then, $\Phi_\ocircle(\SSS)\le \lambda l q$, since each edge in \(G\) can count for $\lambda l$ circle-nodes if the \(d\) corresponding rectangle-nodes propagate all campaigns in $[\MU-\NU+d]$. It follows that $|E(S)| \ge \Phi_\ocircle(\SSS)/(\lambda l)$.\qedhere
  \end{enumerate}
\end{proof}

\section{Deferred Proofs for Section~\ref{section: heterogeneous}}

\subsection{Deferred Proofs for the Analysis of Algorithm \textsc{GreedyTuple}} \label{section: appendix greedy tuple}

The aim of this section is to prove the following Lemma.
\lemEquiv*

For this purpose, we will first prove the following lemma.
\begin{lemma} \label{lem:equiv3}
    Let $0 < \epsilon < 1$, \(\delta\le 1/2\), and $\ell\in [1,\NU - 1]$. With probability at least \(1-\delta\), after each iteration $i$ of \textsc{GreedyTuple}\((\eps, \delta, \ell, \III, \NU, k)\), it either holds that
    \[
        \Phi^{\ge \ell}(\SSS^i) - \Phi^{\ge \ell}(\SSS^{i-1}) \ge \frac{1-\frac{\eps}{2}}{\binom{k}{\NU - \ell}}\cdot(\Phi^{\ge \ell}(\SSS_{\ge \ell}^*) - \Phi^{\ge \ell}(\SSS^{i-1})).
        \quad\text{ or }\quad
        \Phi^{\ge \ell}(\SSS^i) \ge \big(1-\frac{1}{e}\big)\cdot \Phi^{\ge \ell}(\SSS_{\ge \ell}^*).
    \]
\end{lemma}
\begin{proof}
    Algorithm \textsc{GreedyTuple}\((\eps, \delta, \ell, \III, \NU, k)\) calls algorithm \texttt{approx} at most $t$ times.
    Let us call $E_{i}$ the event that the $i$'th call to \texttt{approx} ``succeeds'', i.e., that the call results in $1\pm\eps'$-approximation $\tilde\Phi^{\ge \ell}(\TTT)$. That is, it holds that $(1-\eps')\Phi^{\ge \ell}(\TTT)\le\tilde\Phi^{\ge \ell}(\TTT)\le (1+\eps')\Phi^{\ge 1}(\TTT)$.
    This event happens with probability at least $1-\delta'=1-\delta/t$. Since there are at most $t$ many evaluations, using the union bound, we obtain that the probability that all evaluations succeed is at least $1-\delta$.
    Now the statement follows with Lemma~\ref{lemma:approximate difference}.
    It states that either \(\Phi^{\ge \ell}(\SSS^i) \ge \big(1-\frac{1}{e}\big)\cdot \Phi^{\ge \ell}(\SSS_{\ge \ell}^*)\) or, for the element \(\tau\) picked by the algorithm,
    it holds that \(\Phi^{\ge \ell}(\SSS^{i-1}\cup\tau)-\Phi^{\ge \ell}(\SSS^{i-1}) \ge (1-\frac{\eps}{2})/\binom{k}{\NU - \ell}\cdot(\Phi^{\ge \ell}(\SSS_{\ge \ell}^*)-\Phi^{\ge \ell}(\SSS^{i-1}))\) using Lemma~\ref{lemma:exists element}.
\end{proof}

We can now prove Lemma~\ref{lem:equiv1}.
\begin{proof}[Proof of Lemma~\ref{lem:equiv1}]
    We show the statement by induction. For $i=1$, we note that by Lemma~\ref{lem:equiv3},
    we either have $\Phi^{\ge \ell}(\SSS^1) - \Phi^{\ge \ell}(\SSS^0) \ge (1-\frac{\eps}{2})/\binom{k}{\NU - \ell}\cdot (\Phi^{\ge \ell}(\SSS_{\ge \ell}^*)-\Phi^{\ge \ell}(\SSS^0))$ or
    $\Phi^{\ge \ell}(\SSS^i) \ge \big(1-\frac{1}{e}\big)\cdot \Phi^{\ge \ell}(\SSS_{\ge \ell}^*)$.
    In the latter case the statement holds, in the former case, we get $\Phi^{\ge \ell}(\SSS^1) \ge (1-\frac{\eps}{2})/\binom{k}{\NU - \ell} \cdot \Phi^{\ge \ell}(\SSS_{\ge \ell}^*)$ and thus the statement follows in both cases.
    For $i>1$, let us assume that the statement holds after iteration $i-1$.
    If $\Phi^{\ge \ell}(\SSS^{i-1}) \ge (1-1/e)\cdot\Phi^{\ge \ell}(\SSS_{\ge \ell}^*)$, we have $\Phi^{\ge \ell}(\SSS^{i}) \ge (1-1/e)\cdot \Phi^{\ge \ell}(\SSS_{\ge \ell}^*)$ by monotonicity.
    In the other case, we have that
    \begin{align}\label{formula:ih}
        \Phi^{\ge \ell}(\SSS^{i-1}) \ge \Big(1-\Big(1-\frac{1-\frac{\eps}{2}}{\binom{k}{\NU-\ell}}\Big)^{i-1}\Big) \cdot \Phi^{\ge \ell}(\SSS_{\ge \ell}^*).
    \end{align}
    Applying Lemma~\ref{lem:equiv3} yields that either $\Phi^{\ge \ell}(\SSS^i) \ge (1-1/e)\cdot \Phi^{\ge \ell}(\SSS_{\ge \ell}^*)$, in which case the statement holds, or we obtain
    \begin{align*}
        \!\!\Phi^{\ge \ell}(\SSS^i) &\!=\! \Phi^{\ge \ell}(\SSS^{i-1}) \!+\! (\Phi^{\ge \ell}(\SSS^i) \!-\! \Phi^{\ge \ell}(\SSS^{i-1}))
        \!\ge\! \Big(1\!-\!\frac{1\!-\!\frac{\eps}{2}}{\binom{k}{\NU-1}}\Big) \Phi^{\ge \ell}(\SSS^{i-1})\!+\!\frac{1\!-\!\frac{\eps}{2}}{\binom{k}{\NU-1}}\Phi^{\ge \ell}(\SSS_{\ge \ell}^*).
    \end{align*}
    Applying~\eqref{formula:ih} yields the claim.
\end{proof}

\subsection{Deferred Proofs for the Analysis of Algorithm \textsc{GreedyIter}} \label{section: appendix greedy iter}
In this section we prove lemmata \ref{lem:equiv2} and \ref{lem:recursion greedy 2} which are paramount in proving the approximation ratio of Algorithm \textsc{GreedyIter}.
\lemTwoEquiv*
\begin{proof}
  Let $\SSS_{\lfloor k/(\NU-1) \rfloor}^*$ and $\SSS_{k}^*$ be  sets of cardinality $\lfloor k/(\NU-1) \rfloor$ and $k$, respectively, maximizing $\Phi^{\ge 1}_2(\III, \cdot)$.
  Furthermore, let $\TTT$ be a subset of $\SSS_{k}^*$ of cardinality $\lfloor k/(\NU-1)\rfloor$ that maximizes $\Phi^{\ge 1}_2(\III,\cdot)$. Lemma~\ref{lemma:standard greedy} yields that for $\eps'=\eps/2$, with probability at least $1-\delta/\NU$, we have that
  \begin{align}\label{formula:uppper bound T}
    \Phi^{\ge 1}_2(\III, \SSS^{[1]})
    \ge \big(1- \frac{1}{e} - \eps'\big)\cdot
    \Phi^{\ge 1}_2(\III, \SSS_{\lfloor k/(\NU-1) \rfloor}^*)
    \ge \big(1- \frac{1}{e} - \eps'\big)\cdot
    \Phi^{\ge 1}_2(\III, \TTT).
  \end{align}
  Using the submodularity and monotonicity of $\Phi^{\ge 1}_2(\III,\cdot)$ and the maximum choice of $\TTT$ yields
  \begin{align}\label{formula:lower bound T}
    \Phi^{\ge 1}_2(\III,\SSS_{k}^*)
    \le \left \lceil\frac{k}{\lfloor k/(\NU-1) \rfloor}\right\rceil\cdot
    \Phi^{\ge 1}_2(\III,\TTT)
    \le
    \frac{\NU}{1-\eps'}\cdot
    \Phi^{\ge 1}_2(\III,\TTT),
  \end{align}
  as $k\ge (\NU - 1)/\eps'$ implies $k - \NU +1 \ge (1-\eps')k$ and thus
  $
    \lceil\frac{k}{\lfloor k/(\NU-1) \rfloor}\rceil
    \le \lceil\frac{k(\NU-1)}{k - \NU +1}\rceil
    \le \lceil\frac{\NU-1}{1-\eps'}\rceil
    \le \frac{\NU}{1-\eps'}
  $.
  By combining the estimates from~\eqref{formula:uppper bound T} and~\eqref{formula:lower bound T}, we obtain
  \begin{align*}
    &\Phi^{\ge 1}_2(\III,\SSS^{[1]})
    \ge \frac{(1-\frac{1}{e}-\eps')(1-\eps')}{\NU} \cdot\Phi^{\ge 1}_2(\III,\SSS_{k}^*)
    \ge \frac{(1-\frac{1}{e}-\eps)}{\NU} \cdot \Phi^{\ge 1}_\NU(\III, \SSS_{\ge 1}^*),
  \end{align*}
  where the last step uses that $x(1-\eps')\ge x-\eps'$ for any $x\le 1$, the definition of $\eps'=\eps/2$, and the fact that
  $\SSS_{k}^*$ and $\SSS_{\ge 1}^*$ are both of size $k$ and thus
  $\Phi^{\ge 1}_2(\III,\SSS_{k}^*)\!\ge\! \Phi^{\ge 1}_2(\III,\SSS_{\ge 1}^*)\!\ge\! \Phi^{\ge 1}_\NU(\III, \SSS_{\ge 1}^*)$.
\end{proof}

\lemRecursionGreedyTwo*
\begin{proof}
  We use the shorthand $\Phi^{[\ell]}(\cdot)\!:=\!\Phi^{\ge \ell}_{\ell+1}(\RRR^{[\ell]}, \cdot)$ and similar $\Phi^{[\ell-1]}(\SSS)\!:=\!\Phi^{\ge \ell -1}_{\ell}(\RRR^{[\ell-1]}, \cdot)$.
  We define $U:=V\times[\ell+1]$ and partition it into sets of cardinality $\lfloor k/(\NU-1) \rfloor$ plus a possible set of smaller size. The number of sets in the partition is $t:=\lceil\frac{(\ell+1)|V|}{\lfloor k/(\NU-1) \rfloor}\rceil$. Denote these sets by $U_1,\ldots, U_t$.
  Now, let $\TTT$ be any set of cardinality $\lfloor k/(\NU-1) \rfloor$ that maximizes $\Phi^{[\ell]}(\cdot)$ and assume for the purpose of contradiction that
  \begin{align}\label{formula:TTT small}
    \Phi^{[\ell]}(\TTT) - \Phi^{[\ell]}(\emptyset)
    < \frac{1}{t}\cdot (\Phi^{[\ell - 1]}(\SSS^{[\ell-1]}) - \Phi^{[\ell]}(\emptyset)).
  \end{align}
  By definition of $\TTT$, we have $\Phi^{[\ell]}(\TTT)\ge \Phi^{[\ell]}(U_i)$ for $i\in[t]$.
  Hence, by submodularity we get
  \begin{align*}
    \Phi^{[\ell]}(U) \!-\! \Phi^{[\ell]}(\emptyset)
    \!\le\! \sum_{i=1}^t \!\big(\Phi^{[\ell]}(U_i) \!-\! \Phi^{[\ell]}(\emptyset)\big)
    \!\le\! t \big(\Phi^{[\ell]}(\TTT) \!-\! \Phi^{[\ell]}(\emptyset)\big)
    \!< \!\Phi^{[\ell - 1]}(\SSS^{[\ell-1]}) \!-\! \Phi^{[\ell]}(\emptyset).
 \end{align*}
 Since the maximum possible number of nodes, say $N$ are guaranteed to be reached by $\ell+1$ campaigns from sets $U$, we have however that $\Phi^{[\ell]}(V^{\ell+1})=N$.
 On the other hand we have $\Phi^{[\ell - 1]}(\SSS^{[\ell-1]})\le N$, which leads to a contradiction.
 From Lemma~\ref{lemma:standard greedy} we know that, with probability at least $1-\delta/\NU$, it holds that $\Phi^{[\ell]}(\SSS^{[\ell]})\ge (1 - 1/e - \eps')\cdot\Phi^{[\ell]}(\TTT)$ with $\eps'=\eps/2$.
 Thus, together with the converse of~\eqref{formula:TTT small}, we get
 \begin{align*}
   \Phi^{[\ell]}(\SSS^{[\ell]})
   &\ge \frac{1-\frac{1}{e}-\eps'}{t}\cdot(\Phi^{[\ell - 1]}(\SSS^{[\ell-1]}) - \Phi^{[\ell]}(\emptyset)) + \Phi^{[\ell]}(\emptyset)
   \ge \frac{1-\frac{1}{e}-\eps'}{t}\cdot\Phi^{[\ell - 1]}(\SSS^{[\ell-1]}).
 \end{align*}
 It remains to observe that $k\ge (\NU - 1)/\eps'$ implies $k - \NU +1 \ge (1-\eps')k$ and thus
 \begin{align*}
   t
   = \left\lceil\frac{(\ell+1)|V|}{\lfloor k/(\NU-1) \rfloor}\right\rceil
   \le \left\lceil\frac{(\ell+1)(\NU-1)|V|}{k - \NU + 1}\right\rceil
   \le \frac{2(\ell+1)(\NU-1)|V|}{(1-\eps')k}
 \end{align*}
 where the last inequality follows since $k\le \NU \cdot |V|$ and $\NU\ge 2$ yield that the argument of the ceil-function is at least 1, and thus the error due to rounding is upper bounded by a factor of 2. The choice of $\eps'=\eps/2$ leads the result.
\end{proof}

\section{Deferred Proofs for Section~\ref{section: correlated}}\label{appendix: correlated}

\correlatedclaims*
\begin{proof}
  \begin{enumerate}
        \item  Define $\TTT:= \{(v,0)| (v,i) \in \SSS^k\}$ and observe that $|\TTT|\le k$. For a given outcome $\XXX$ a node that contributes to $\Phi^{\ge 1}(\SSS^k)$ is either reached by at least $\NU$ campaigns in $\III$ or has to be reached by a node in $\SSS^k$.
        In this case, for the same $\XXX$ this node will also contribute to $\Psi(\TTT)$. Hence, we have $\Psi(\TTT) \ge \Phi^{\geq 1}(\SSS^k)$. The optimality of $\TTT^k$ concludes the proof.
        \item First observe that $\lceil \frac{k}{\lfloor k/\NU \rfloor} \rceil < \frac{k}{k/\NU - 1} + 1 \leq \frac{\NU+1}{1-\eps}$ by the assumption on $k$. Now, let $\TTT$ be a subset of $\TTT^k$ of size $\lfloor k/\NU \rfloor$ maximizing $\Psi$. By submodularity of $\Psi$, we have $\Psi(\TTT^k) \le \lceil \frac{k}{\lfloor k/\NU\rfloor} \rceil \Psi(\TTT)\le \frac{\NU+1}{1-\eps} \Psi(\TTT)$. Using the optimality of $\TTT^{\lfloor k/\NU\rfloor}$ concludes the proof.
        \item Since the cascade processes are completely correlated, given an outcome $\XXX$, assume that a node contributes to $\Psi(\TTT)$, then either it is reached by $\NU$ campaigns from $\III$ or it is reached by $\TTT$. In the former case, the same node also contributes to $\Phi^{\geq 1}(\SSS')$ as it is reached by $\NU$ campaigns from $\III$. In the later case, it will be reached by all campaigns in $[\NU]$ by $\SSS'$ and will therefore also contribute to $\Phi^{\geq 1}(\SSS')$.\qedhere
  \end{enumerate}
\end{proof}

\end{document}